\documentclass[12pt,english,american]{article}

\usepackage{hyperref,amsmath,amsfonts,amssymb,amsthm,mathrsfs,enumerate,bm,bbm,verbatim,textcomp}

\theoremstyle{plain}\newtheorem{theorem}{Theorem}[section]
\theoremstyle{plain}\newtheorem{lemma}[theorem]{Lemma}
\theoremstyle{definition}\newtheorem{definition}[theorem]{Definition}

\theoremstyle{remark}\newtheorem{remark}{Remark} 
\theoremstyle{definition}\newtheorem{def:and:lemma}[theorem]{Definition and Lemma}

\newcommand{\lsp}{\big \langle}
\newcommand{\rsp}{\big \rangle}

\newcommand{\im}{\operatorname{Im}}
\newcommand{\re}{\operatorname{Re}}
\newcommand{\lno}{\Big\vert\! \Big \vert}
\newcommand{\rno}{\Big\vert\! \Big \vert}
\newcommand{\no}{\vert\hspace{-1pt} \vert}

\newcommand{\wt}{\widetilde}

\newcommand{\p}[2]{\ensuremath{p^{\ensuremath{{#2}}}_{#1}}}
\newcommand{\q}[2]{\ensuremath{q^{\ensuremath{{#2}}}_{#1}}}
\newcommand{\cf}[2]{\ensuremath{\widehat {{#1}}^{\ensuremath{{#2}}}}}
\newcommand{\cfs}[3]{\ensuremath{\widehat {{#1}}_{#3}^{\ensuremath{{#2}}}}}

\allowdisplaybreaks[1]

\title{Bogoliubov corrections and trace norm convergence for the Hartree dynamics}

\date{September 19, 2016}

\author{
David Mitrouskas\footnote{Corresponding author.\ Ludwig-Maximilians-Universit\"at, Mathematisches Institut, Theresienstr.\ 39, {80333} M\"unchen, Germany. E-mail: {\tt dmitrous@math.lmu.de}},
S\"oren Petrat\footnote{Institute of Science and Technology Austria (IST Austria), Am Campus 1, 3400 Klosterneuburg, Austria. E-mail: {\tt soeren.petrat@ist.ac.at}},
Peter Pickl\footnote{Ludwig-Maximilians-Universit\"at, Mathematisches Institut, Theresienstr.\ 39, {80333} M\"unchen, Germany. E-mail: {\tt pickl@math.lmu.de}}
}

\begin{document}

\maketitle

\begin{abstract}
We consider the dynamics of a large number $N$ of nonrelativistic bosons in the mean field limit for a class of interaction potentials that includes Coulomb interaction.\ In order to describe the fluctuations around the mean field Hartree state, we introduce an auxiliary Hamiltonian on the $N$-particle space that is very similar to the one obtained from Bogoliubov theory. We show convergence of the auxiliary time evolution to the fully interacting dynamics in the norm of the $N$-particle space. This result allows us to prove several other results:\ convergence of reduced density matrices in trace norm with optimal rate, convergence in energy trace norm, and convergence to a time evolution obtained from the Bogoliubov Hamiltonian on Fock space with expected optimal rate.\ We thus extend and quantify several previous results, e.g., by providing the physically important convergence rates, including time-dependent external fields and singular interactions, and allowing for general initial states, e.g., those that are expected to be ground states of interacting systems.
\end{abstract}

\section{Introduction}

A system of $N$ spinless bosons in nonrelativistic quantum mechanics is described by a wave function $\Psi \in L^2_s(\mathbb R^{3N})$, the subspace of square integrable functions that are symmetric under permutations of the variables $x_1,...,x_N \in \mathbb R^{3}$ (we only consider three dimensions here). We always assume that $\Psi$ is normalized, i.e., $\no \Psi \no = 1$, such that $\vert \Psi(x_1,...,x_N)\vert^2$ can be interpreted as the probability density of finding particle one at position $x_1$, particle two at position $x_2$, and so on. The time evolution of the wave function is governed by the nonrelativistic many-body Schr\"odinger equation
\begin{align}
\label{S:E}& i\partial_t \Psi_t =  H^t  \Psi_t,
\end{align}
where the Hamiltonian operator $H^t$ is of the form
\begin{align}
H^t = \sum_{i=1}^N h_{i}^t + \lambda_N \sum_{1\le i < j \le N} v(x_i-x_j). \label{HAMILTONIAN:H}
\end{align}
Here, $h_{i}^t=-\Delta_i + W_{i}^t$ denotes a one-particle operator, $\Delta_i$ is the Laplacian describing the kinetic energy of the $i$-th particle, $W_{i}^t=W^t(x_i)$ a possibly time-dependent external potential. The interaction between the gas particles is described by a real-valued function $v=v(x)$, e.g., the Coulomb potential $v(x)=1/\vert x\vert$. The coupling constant in front of the interaction will be chosen as $\lambda_N=1/(N-1)$ which ensures that the average interaction energy is of the same order as the kinetic energy, namely of order $N$. In this situation, a nontrivial behavior of the many-body system can be expected for large particle number $N$. Our goal in this work is to investigate the large $N$ limit of solutions to the Schr\"odinger equation and, in particular, the corrections to the leading-order mean field component of such solutions.

The physical setting we have in mind is that the gas is initially trapped in a confining potential $W^0$ and cooled down, such that $\Psi_0$ is close to the ground state of $H^0$. By removing or changing the external field $W^t$, the ground state of $H^0$ is in general not an eigenfunction of $H^t$ for $t>0$ anymore, so the time evolution is nontrivial. To our understanding, this is the picture behind the experiments of ultracold gases exhibiting the phenomenon of Bose-Einstein condensation, see, e.g., \cite{BlochDalibardZwerger:2008} and references therein.\footnote{Note, however, that for actual experiments the Gross-Pitaevskii limit is more relevant, which is more involved than the mean field limit we are considering in the present work.}\\

It has been established in many different settings that Hartree theory emerges as the macroscopic description of the low temperature many-body Bose gas in the mean field regime, i.e., for $N\to \infty$, $N \lambda_N \to 1$. Hartree theory is defined by the one-body Hamiltonian 
\begin{align}
 h^{t,\varphi}= h^t + v\ast \vert \varphi \vert^2 - \mu^\varphi, \hspace{0.5cm} \varphi \in L^2(\mathbb R^3),\label{HARTREE:HAMILTONIAN}
\end{align}
where $\ast$ denotes the convolution of functions on $\mathbb R^3$, and the phase factor $\mathbb R \ni \mu^\varphi = \frac{1}{2} \int dx \ \big(v\ast \vert \varphi \vert^2\big)(x) \vert \varphi(x)\vert^2$ is chosen for later convenience. In order to understand the relation between the microscopic model defined by \eqref{HAMILTONIAN:H} and Hartree theory, think of a completely factorized $N$-particle wave function $\Psi = \varphi^{\otimes N}$ for which the potential term in $H$ corresponds to a sum of identically and independently distributed random variables with probability density $\vert \varphi(x) \vert^2$. It follows from the law of large numbers that the potential felt by, e.g., the first particle, at position $x$ is given by
$$ \frac{1}{N-1} \sum_{i=2}^N v(x_i-x) \approx \frac{1}{N-1}  {\sum_{i=2}^N} \int  v(x_i-x) \vert \varphi (x_i) \vert^2 dx_i = \big(v\ast  \vert \varphi \vert^2\big) (x).$$
The $N$-particle Hamiltonian $H$ is hence expected to act as a sum of $N$ one-body Hamiltonians each given by \eqref{HARTREE:HAMILTONIAN}. In more precise terms, the Hartree Hamiltonian governs the leading order dynamics of a wave function which is initially close to a condensate $\varphi_0^{\otimes N}$, in the sense of
\begin{align}
\lim_{N\to \infty} \text{Tr}\big\vert  \gamma_{\Psi_0}^{(1)} - \vert \varphi_0 \rangle \langle \varphi_0 \vert \big\vert = 0 ~~~\Rightarrow~~~ \lim_{N\to \infty} \text{Tr}\big\vert  \gamma_{\Psi_t}^{(1)} - \vert \varphi_t \rangle \langle \varphi_t \vert \big\vert =0, \label{PROPAGATION:OF:CHAOS}
\end{align}
where $\text{Tr}$ denotes the trace, and the one-particle state $\varphi_t$ solves the nonlinear time-dependent Hartree equation
\begin{align}
i\partial_t \varphi_t =  h^{t,\varphi_t} \varphi_t \label{TIME:DEP:HARTREE:EQUATION}
\end{align}
with initial condition $\varphi_0$. The operator $\gamma_{\Psi}^{(1)}: L^2(\mathbb R^3) \to L^2(\mathbb R^3) $ is the one-body reduced density matrix of $\Psi\in L^2_s(\mathbb R^{3N})$, defined by its kernel
\begin{align*}
\gamma_{\Psi}^{(1)}(x,y) = \int  \Psi (x,x_2,...,x_N) \overline{ \Psi(y,x_2,...,x_N)}  dx_2...dx_N, \hspace{0.5cm} \text{and} \hspace{0.5cm} \vert \varphi_t \rangle \langle \varphi_t \vert= \gamma^{(1)}_{\varphi_t^{\otimes N}}.
\end{align*}
Implications like \eqref{PROPAGATION:OF:CHAOS} are referred to as propagation of chaos or persistence of condensation, and have been proven in different and very general settings, e.g., \cite{hepp:1974,GinibreVelo:1979a,
GinibreVelo:1979b,Spohn:1980,ErdosYau:2001,
RodnianskiSchlein:2007,
ElgartSchlein:2007,ErdosSchlein:2008,
FrohlichKnowlesSchwarz:2009,ChenLee:2010,pickl:2011method,
picklknowles:2010,Chen:2011}.
The question whether and for which situations factorization holds in the first place is answered by Hartree theory as well. The ground state of a weakly interacting Bose gas obeys the property of Bose-Einstein condensation: the wave function $\Psi^{(0)}$ corresponding to the lowest eigenvalue of a Hamiltonian of the form \eqref{HAMILTONIAN:H} factorizes into an $N$-fold product of a single one-particle wave function $\varphi^{(0)}$ which is determined by minimizing the nonlinear Hartree functional 
\begin{align}\label{HARTREE:FUNCTIONAL}
 \mathcal E_{h^{0,\varphi}}(\varphi) = \{ \langle \varphi, h^{0,\varphi} \varphi \rangle : \varphi \in H^1(\mathbb R^3), \no \varphi \no =1 \}.
\end{align}
The condensation property holds again in the reduced sense (being true at least in the case when there exists a unique minimizer of $\mathcal E^{h^{0,\varphi}}(\varphi)$), i.e.,
\begin{align}
\lim_{N\to \infty} \text{Tr}\big\vert  \gamma_{\Psi^{(0)}}^{(1)} - \vert \varphi^{(0)} \rangle \langle \varphi^{(0)} \vert \big\vert =0. \label{CONDENSATION}
\end{align}
It further holds for a comparison of the energies, since $E^{(0)}_N = Ne^{(0)} +o_N(N)$ where $E^{(0)}_N$ denotes the infimum of the spectrum $\sigma(H^0)$ and $e^{(0)}$ the infimum of \eqref{HARTREE:FUNCTIONAL} (the symbol $o_N(A)$ stands for terms with $o_N(A)/A \to 0$ for $N\to\infty$). The rigorous analysis of this question goes back to \cite{BenguriaLieb:1983,LiebYau:1987}. For recent results and an extensive list of references, we refer to \cite{LewinNamRougerie:2014}.

The notion of distance in \eqref{PROPAGATION:OF:CHAOS} and \eqref{CONDENSATION} is equivalent to convergence of bounded $k$-particle operators with norm of order one (for fixed $k$ when $N$ tends to $\infty$). This, in turn, is strong enough to imply a law of large numbers type result for such observables. In order to control unbounded observables, e.g., the energy or momentum, a slightly stronger statement than \eqref{PROPAGATION:OF:CHAOS} is needed. A suitable generalization is given by
\begin{align}\label{SOBOLEV:TRACE:NORM}
\lim_{N\to \infty} \text{Tr}\big\vert \sqrt{1-\Delta} \big( \gamma_{\Psi_t}^{(1)} - \vert \varphi_t \rangle \langle \varphi_t \vert \big) \sqrt{1-\Delta} \big\vert =0,
\end{align}
i.e., convergence in the so-called energy trace norm. Questions in this direction have been studied in \cite{Michelangeli:2012,Luhrmann:2012} and more recently in \cite{anapolitanos:2016}. Yet another natural notion of distance, much stronger compared to convergence in terms of reduced densities, is the $L^2$-norm on the full $N$-particle space $L^2_s(\mathbb R^{3N})$.\ In the interacting case, i.e., for $v\neq 0$, the ground state is not close to a product of one-particle wave functions, and neither does the initial product structure survive the dynamics in the $L^2$ sense. If only a single particle is not in the correct condensate wave function, there is no closeness in $L^2$-norm. On the other hand, condensation is a macroscopic phenomenon, i.e., it still holds, even if a few out of a very large number $N$ of particles are not in the condensate.\ The property of condensation, and persistence of condensation, is therefore correctly understood by means of the topology of reduced densities, e.g., in the sense of \eqref{PROPAGATION:OF:CHAOS} or \eqref{SOBOLEV:TRACE:NORM}.\ An approximation in terms of the $L^2$-distance of $\Psi_t$ or the ground state $\Psi^{(0)}$ is nevertheless highly interesting. For example, a large $N$ approximation in $L^2(\mathbb R^{3N})$ is closely connected to the analysis of low energy excitations which are relevant, e.g., for the explanation of superfluidity and other collective phenomena. It can be understood as the next-to-leading order correction to Hartree theory, and is known under the name of Bogoliubov theory \cite{Bogoliubov:1947}. 

The rigorous analysis of spectral low energy properties in terms of Bogoliubov theory for the weakly interacting Bose gas has been initiated more recently. In \cite{GiulianiSeiringer:2009,LiebSolovej:2001,
LiebSolovej:2004,Solovej:2006,YauYin:2009}, the next-to-leading order contribution $E^{Bog}$ in the ground state energy $E_N^0 = Ne^{(0)} + E^{Bog} + o_N(1)$ has been derived. Then, in \cite{Seiringer:2011}, the complete Bogoliubov theory (of the low energy spectrum and low energy eigenfunctions) was derived for the homogenous gas on the torus, which was generalized in \cite{Grech:2013,Lewin:2015b}, and further generalized to a mean field large volume limit in \cite{DerezinskiNapiorkowski:2014}. In \cite{Lewin:2015a}, Bogoliubov theory was derived also for the time-dependent problem.\ The main result of \cite{Lewin:2015a} is a full characterization of fluctuations in $\Psi_t$ around the Hartree product $\varphi_t^{\otimes N}$. It was shown that
\begin{align} \label{TIME:DEP:FLUCTUATIONS}
\lim_{N\to \infty} \lno \Psi_t - \sum_{k=0}^N \varphi^{\otimes N-k}_t \otimes_s \chi^{(k)}_t \rno = 0,
\end{align}
where $\otimes_s$ stands for the normalized symmetric tensor product, cf.\ \eqref{SYMMETRIC:PRODUCT}, and the correlation functions $(\chi^{(k)}_t)_{k\ge 0} $ solve a Schr\"odinger equation on the bosonic Fock space (constructed over orthogonal complement of the Hartree solution $\{ \varphi_t \} \subset L^2(\mathbb R^{3N})$) with $N$-independent, quadratic Hamiltonian.\\

In this work, we contribute to the understanding of \eqref{PROPAGATION:OF:CHAOS}, \eqref{SOBOLEV:TRACE:NORM} and \eqref{TIME:DEP:FLUCTUATIONS}. Besides that, we introduce a first quantized version of Bogoliubov theory. Our strategy is to first show norm convergence of $\Psi_t$ towards the solution $\widetilde \Psi_t$ of a Schr\"odinger equation with a modified and much simpler quadratic Hamiltonian $\wt H_t$. If we denote by $\p{i}{t} = \vert \varphi_t(x_i) \rangle \langle \varphi_t(x_i) \vert$ respectively by $\q{i}{t}=1-\p{i}{t}$ the orthogonal projector in the variable $x_i$ onto the subspace $\{ \varphi_t\}$ respectively its orthogonal complement, the Hamiltonian $\wt H^t$ is obtained from the original Hamiltonian \eqref{HAMILTONIAN:H},
\begin{align*}
H^t = & \sum_{i=1}^N  h_{i}^{t,\varphi_t} + \lambda_N \sum_{1\le i<j\le N}  (\p{i}{t} + \q{i}{t}) (\p{j}{t} + \q{j}{t}) \Big(  v_{ij}- v\ast \vert \varphi_t \vert_i^2 + \mu^{\varphi_t} \Big) (\p{i}{t} + \q{i}{t}) (\p{j}{t} + \q{j}{t}),
\end{align*}
by discarding all terms that contain three or four $q^t$'s.\ When acting on wave functions that are, in a certain sense, sufficiently close to the Hartree product, $\widetilde H^t$ coincides, up to terms that become small for large $N$, with the usual definition of the Bogoliubov Hamiltonian.\ After showing in Theorem \ref{THEOREM:MAIN:THEOREM:NORM:APPROXIMATION} that
\begin{align*}
\no \Psi_t - \widetilde \Psi_t \no \le \frac{ e^{C^{\varphi_t}}}{\sqrt N},
\end{align*}
where $\wt \Psi_t$ solves the Schr\"odinger equation with Hamiltonian $\widetilde H^t$, we derive
\begin{enumerate}
\item in Theorem \ref{CONVERGENCE:OF:DENSITY} that $\gamma^{(1)}_{\Psi_t}$ converges to $\vert \varphi_t \rangle \langle \varphi_t \vert$ in trace norm with optimal rate $1/N$, as well as in energy trace norm with rate $1/\sqrt N$,
\item in Theorem \ref{THM:BOG:HIERARCHY} the approximation of $\Psi_t$ in terms of correlation functions $(\chi_t^{(k)})_{k\ge 0}$ in the sense of \eqref{TIME:DEP:FLUCTUATIONS}, with expected optimal rate.
\end{enumerate}
Both results hold for a class of initial states that is expected to include ground states of interacting systems. As already noted above, those initial states have a nontrivial dynamics if external fields are removed or changed. The first point is an extension of earlier results. Convergence in trace norm with rate $1/N$ was shown in \cite{ErdosSchlein:2008,ChenLee:2010,Chen:2011,Lee:2013} for initial wave functions that are completely factorized. It is interesting to note, and this is in agreement with the previous results, that the Hartree approximation is not sufficient to prove \eqref{PROPAGATION:OF:CHAOS} with optimal rate. In \cite{anapolitanos:2016}, convergence in terms of the energy trace norm was derived by means of a compactness argument without explicit error. The characterization of $\Psi_t$ in terms of correlation functions $(\chi_t^{(k)})_{k\ge 0}$, was first studied in a mathematical rigorous way in \cite{Lewin:2015a}.\ In the present work, we derive the optimal error of this approximation which was not included in the analysis of \cite[cf. Remark 3]{Lewin:2015a}.\\

From the technical point of view, our approach consists of a generalization of the method that was used to derive the Hartree equation in \cite{pickl:2011method,picklknowles:2010}. We expect this approach to turn out stable and versatile and therefore also useful in order to derive similar results for more complicated situations, in particular for the NLS and Gross-Pitaevskii limit. Next-to-leading order corrections in $\Psi_t$ have been studied very recently in \cite{nam:2015,nam:2016} for the NLS equation with $0\le \beta <\frac{1}{2}$. Let us note that there are also very strong results about the $L^2$-approximation for states on Fock space, derived by means of the coherent state method that goes back to Hepp \cite{hepp:1974}. These results cover the weakly interacting case \cite{GrillakisMachedon:2013} as well as the NLS limit for all $\beta <1$ \cite{Boccato:2016} (a detailed list of references can be found in \cite[Section 1.2]{nam:2015}). The intial states are here coherent states in Fock space, or slight generalizations thereof. These results also give convergence in $L^2$-norm for initial $N$-particle product states, with a worse convergence rate. It is unclear whether the results also imply convergence of initial $N$-particle ground states, like we consider in this article, or that were considered in \cite{Lewin:2015a,nam:2015,nam:2016}. For a more detailed comparison, we refer to \cite[Section 3]{Lewin:2015a}.

\section{Main results}

The basic assumptions throughout our analysis are summarized in\\
\\
\textbf{Assumptions A.1.}
\begin{enumerate}
\item The external potential satisfies $W^t\in L^\infty(\mathbb R^3)$ and $\partial_t W^t \in L^\infty (\mathbb R^3)$ for all $t\ge 0$.
\item The two-particle potential $v$ is real valued and even and satisfies $v^2 \le C(1-\Delta)$.
\end{enumerate}
It follows from standard arguments that $H^t$, $t \ge 0$, is self-adjoint on $H_s^2(\mathbb R^{3N})$ and that the $N$-body propagator $U_{t-s}$, defined by $i\partial_t U_{t-s} = H^t U_{t-s}$, $U_0 = 1$, exists and satisfies $U_t \Psi_0 \in H_s^2(\mathbb R^{3N})$ for all $t$. Invoking the assumed bound on $v$, Hardy's inequality yields
\begin{align}\label{BOUNDS:ON:V:PHI}
\no v \ast \vert \varphi_t \vert^2 \no_{\infty} + \no v^2 \ast \vert \varphi_t \vert^2 \no_{\infty}  \le   C ( \no \varphi_t \no \ \no \nabla\varphi_t \no + \no \nabla\varphi_t \no^2 ) .
\end{align}
It further follows (e.g., by adapting the techniques used in \cite{chadam:1975} to our setting) that for every initial wave function $\varphi_0 \in H^2(\mathbb R^3)$, the Hartree equation admits a unique solution 
$$ \varphi_t \in  C ( [0,\infty ) ,  H^2(\mathbb R^3) ) \cap  C^1( [0,\infty) , L^2(\mathbb R^3) ).$$
\begin{remark}
A model that satisfies all listed assumptions and which we recommend to have in mind is the Bose gas with Coulomb interaction, $h = -\Delta$, $v (x) = \pm  \vert x \vert^{-1}$, where the initial trap is removed at time zero. 
\end{remark} 

We now define some operators we will use throughout this article. Let $\varphi_t$ denote the solution to the Hartree equation.
\begin{definition}\label{DEF:PROJECTORS}
For any $1\le i \le N$, we define the time-dependent projectors
\begin{align*}
\p{i}{t}:& L^2(\mathbb R^{3N}) \to L^2(\mathbb R^{3N}), \hspace{1cm} \p{k}{t} \Psi(x_1,...,x_N) =    \varphi_t (x_i) \int \overline {\varphi_t(x_i)} \Psi(x_1,...,x_N) dx_k,
\end{align*}
and $q^t_k = 1-p_k^t$.
\end{definition}
\begin{definition}
\label{DEF:H:TILDE}
We introduce for all $t\ge 0$ the Bogoliubov Hamiltonian
\begin{align}
\widetilde H^t = \sum_{i=1}^N  h^{t,\varphi_t}_{i} + \lambda_N \sum_{1\le i<j\le N} \Big[ \Big( \p{i}{t}\q{j}{t} v(x_i-x_j) \q{i}{t}\p{j}{t} + \p{i}{t}\p{j}{t} v(x_i-x_j) \q{i}{t}\q{j}{t} \Big) + \text{h.c.} \Big]		,
\end{align}
where h.c.\ stands for the hermitian conjugate of the preceding terms. $\widetilde H^t$ is self-adjoint on $H^2_s(\mathbb R^{3N})$ for all $t\ge 0$. We denote the corresponding unitary time evolution by $\wt \Psi_t = \widetilde U_t \Psi_0$ for any $\Psi_0 \in L^2(\mathbb R^{3N})$. A straightforward computation verifies the identity
\begin{align}\label{DIFFERENCE:HAMILTONIAN}
H = \wt H^t + \lambda_N \sum_{1\le i \neq j \le N} \Big[ \Big( \q{i}{t}   \q{j}{t}  \big( v_{ij}- \bar v^{\varphi_t}_j \big) \q{i}{t} \p{j}{t}   +  \q{i}{t} \q{j}{t} \big( v_{ij} - \bar v^{\varphi_t}_j +\mu^{\varphi_t} \big) \q{i}{t} \q{j}{t} \Big) + \text{h.c.} \Big]
\end{align}
where $v_{ij} = v(x_i-x_j)$, $\bar v^{\varphi_t}_j = \big( v\ast \vert \varphi_t\vert^2 \big) (x_j)$ and $\mu^{\varphi_t} = \frac{1}{2} \int  \big( v\ast \vert \varphi_t\vert^2 \big) (x) \vert \varphi_t(x)\vert^2 dx$.
\end{definition}

\noindent \textbf{Notation.}\ We use $\langle \cdot , \cdot \rangle$ for the scalar product and $\no \cdot \no$ for the corresponding norm for both, the $N$-particle and one-particle space, i.e., for $L^2(\mathbb R^{3N})$ and $L^2(\mathbb R^{3})$, respectively.\ The letter $C$ denotes a constant whose value may
change from one line to another.\ In particular, all constants are independent of $N$.

\subsection{Norm convergence of $\Psi_t$ towards $\widetilde \Psi_t$}

The main ingredient for our first main result are the following estimates for $\Psi_t$ and $\wt \Psi_t$ (all proofs are postponed to Section \ref{sec:Proofs}).
\begin{lemma}\label{A:PRIORI:ESTIMATE:PSI} Let $\varphi_t$ be the unique solution of the Hartree equation \eqref{TIME:DEP:HARTREE:EQUATION} with initial condition $\varphi_0 \in H^1(\mathbb R^3)$, $\no \varphi_0 \no = 1$, and let $\Psi_0 \in L^2_s(\mathbb R^{3N})$, $\no \Psi_0 \no=1$. Then, for $\Phi_t\in \{U_t \Psi_0, \widetilde U_t \Psi_0 \}$, there exist positive constants $C^{\varphi_t}, C_n$ such that for all $n\le N$ and all $t>0$,
\begin{align*}
\lsp    \Phi_t, \Big(\prod_{j=1}^n\q{j}{t} \Big) \Phi_t \rsp \le & e^{C^{\varphi_t} } \sum_{i=0}^n \frac{C_n}{N^{n-i}} \lsp \Psi_0, \Big(\prod_{j=1}^i \q{j}{0}  \Big)   \Psi_0 \rsp.
\end{align*}
\end{lemma}

\begin{remark} As \vspace{-0.07cm}a rule of thumb, the lemma tells us that (for appropriate initial wave functions) each projector $q^t$ multiplied with $\Psi_t=U_t\Psi_0$ or with $\wt \Psi_t= \widetilde U_t \Psi_0$ gives a factor $1/\sqrt{N}$.
\end{remark}
\begin{remark} The quantity $\alpha (t) = \lsp  \Psi_t,  \q{1}{t}   \Psi_t \rsp$ counts the relative number of particles outside the condensate in $\Psi_t$. In \cite{pickl:2011method,picklknowles:2010}, it was shown that $\alpha(t) \le e^{C^{\varphi_t}} N^{-1}$ for appropriate initial conditions, which was used to derive $\gamma_{\Psi_t}^{(1)} \to \vert \varphi_t \rangle \langle \varphi_t \vert$ by means of the relation (see, e.g., \cite[Lemma~2.3]{picklknowles:2010})
\begin{align}
\lsp \Psi_t, \q{1}{t} \Psi_t \rsp \le \text{Tr} \big\vert \gamma_{\Psi_t}^{(1)} - \vert \varphi_t\rangle \langle \varphi_t \vert \big\vert \le 
  \sqrt{8 \lsp \Psi_t, \q{1}{t} \Psi_t \rsp }.\label{CONVERGENCE:ALPHA:TRACE}
\end{align}
Along with some modifications, the $n=1$ case has been also used to derive the time-dependent Hartree equation for very singular potentials \cite[Section 4]{picklknowles:2010}, the Hartree equation in a large volume mean field limit \cite{deckert:2014}, the NLS equation without positivity condition \cite{pickl:2010gp_pos} as well as the Gross-Pitaevskii equation \cite{pickl:2010gp_ext}. The idea is also applicable for fermions to derive the Hartree and Hartree-Fock equations in the corresponding mean field limit \cite{Petrat:2016,BachPetrat:2016}.
\end{remark}
 
Each term in the effective two-body potential in $\wt H^t$ contains exactly two $q^t$'s and two $p^t$'s. Therefore, no mass flow is generated between even and odd sectors of the number of particles outside the Hartree state. In order to make this observation precise let us introduce the projectors onto the even and odd sectors in $L^2(\mathbb R^{3N})$ w.r.t.\ the product wave function $\varphi_t^{\otimes N}$,
\begin{align}
\cfs{f}{t}{\text{odd}}: L^2(\mathbb R^{3N}) \to L^2(\mathbb R^{3N}),& \hspace{1cm} \Psi \mapsto \cfs{f}{t}{\text{odd}} \Psi = \sum_{\substack{k=0 \\ k\ \text{odd}}}^N \Big( \prod_{i=1}^k\q{i}{t} \prod_{i=k+1}^N \p{i}{t} \Big)_{sym}\Psi, \label{DEF:F:ODD}\\
\cfs{f}{t}{\text{even}}: L^2(\mathbb R^{3N}) \to L^2(\mathbb R^{3N}),& \hspace{1cm} \Psi \mapsto \cfs{f}{t}{\text{even}} \Psi = \sum_{\substack{k=0 \\ k\ \text{even}}}^N \Big( \prod_{i=1}^k\q{i}{t} \prod_{i=k+1}^N \p{i}{t} \Big)_{sym} \Psi, \label{DEF:F:EVEN}
\end{align} 
where $(\cdot )_{sym}$ abbreviates the symmetric tensor product, cf.\ \eqref{SYMMETRIC:TENSOR:PRODUCT}.\ It follows directly that $1=\cfs{f}{t}{\text{odd}}+\cfs{f}{t}{\text{even}}$ and $\cfs{f}{t}{\text{odd}}\cfs{f}{t}{\text{even}}=\cfs{f}{t}{\text{even}}\cfs{f}{t}{\text{odd}}=0$. 
\begin{lemma}\label{LEMMA:ODD:EVEN:PART:TIME:EVOLUTION}
Let $\varphi_t\in H^1(\mathbb R^3)$ be the solution to the Hartree equation \eqref{TIME:DEP:HARTREE:EQUATION} with initial datum $\varphi_0\in H^1(\mathbb R^3)$, let $\Psi_0 \in L_s^2(\mathbb R^{3N})$ and $\wt \Psi_t = \widetilde U_t \Psi_0$.\ Then
\begin{align*}
\no \cfs{f}{t}{\text{even}}\widetilde \Psi_t \no = \no \cfs{f}{0}{\text{even}}  \Psi_0 \no, \hspace{1cm} \no \cfs{f}{t}{\text{odd}}\widetilde \Psi_t \no = \no\cfs{f}{0}{\text{odd}}  \Psi_0 \no. 
\end{align*}
\end{lemma}
\noindent This property is important in order to obtain bounds on expectation values of one-body operators $A_1$ for which one needs to estimate correlations between the even and odd sector, e.g., terms like $\lsp \Psi, \p{1}{t} A_1 \q{1}{t} \Psi \rsp$. In particular, it is used to control the part of the wave function $\Psi_t$ that describes the few particles outside the Hartree state, see the following Lemma~\ref{LEMMA:ENERGY:ESTIMATE:PSI:TILDE}.
Lemma~\ref{LEMMA:ODD:EVEN:PART:TIME:EVOLUTION} follows easily from verifying that 
$$\Big[\wt H^t - \sum_{i=1}^N  h^{t,\varphi_t}_i , \cfs{f}{t}{\text{odd}} \Big]= 0 = \Big[\wt H^t - \sum_{i=1}^N  h^{t,\varphi_t}_i  ,\cfs{f}{t}{\text{even}} \Big].$$

Since we are after an $L^2$ approximation of $\Psi_t$, it is necessary to have good control of the behavior of all $N$ particles.\ This means that also good control of particles outside the condensate is required.\ For bounded potentials, one can use the fact that the number of such particles is small compared to $N$, and that they can therefore not disturb the other particles too much. For singular potentials, however, already a few badly behaving particles can in principle cause problems when they come close together and generate a large potential energy. That such behavior is very unlikely is due to energy conservation. In order to deal with singular potentials, the idea is thus to use energy conservation to obtain sufficient control of the regularity of the part of $\Psi_t$ which describes particles outside the Hartree product. This, in turn, leads to appropriate bounds on the potential energy of these particles.

\begin{definition}
For $\varphi \in H^1(\mathbb R^3)$, $\Psi\in H^1(\mathbb R^{3N})$, we define the mean field energy and the energy per particle w.r.t.\ $H^t$ by 
\begin{align}
  \mathcal E_{h^{t,\varphi}} : H^1(\mathbb R^3) \to \mathbb R,& \hspace{1cm} \varphi \mapsto  \mathcal E_{h^{t,\varphi}}(\varphi) \hspace{-0.05cm} = \no \nabla_1 \varphi \no^2   + \no W^t_1 \varphi \no^2 + \mu^{\varphi} ,\label{ENERGY:FUNCTIONAL:HARTREE}\\
 \mathcal E_{H^t} : H^1_s(\mathbb R^{3N}) \to \mathbb R,& \hspace{1cm}   \Psi \mapsto \mathcal E_{H^t} (\Psi)\hspace{0.08cm}  = \no  \nabla_1 \Psi \no^2 +  \no  W_1^t \Psi \no^2 
+ \frac{1}{2} \lsp \Psi , v_{12} \Psi \rsp.
\end{align}
\end{definition}

\begin{lemma}
 \label{LEMMA:ENERGY:ESTIMATE:PSI:TILDE} Let $\varphi_t\in H^2(\mathbb R^3)$ be the solution of the Hartree equation \eqref{TIME:DEP:HARTREE:EQUATION} with initial datum $\varphi_0 \in H^2(\mathbb R^3)$, $\no\varphi_0\no =1$, and let $\Psi_0 \in H^{1}_s(\mathbb R^{3N})$, $\no \Psi_0 \no =1$ and $\Psi_t = U_t \Psi_0$, $\widetilde \Psi_t = \widetilde U_t \Psi_0$. Then there exists a positive constant $C^{\varphi_t}$ such that for all $t\ge 0$,
\begin{align*}
 \no \nabla_1 \q{1}{t} \Psi_t \no^2 \le  & C \no \Psi_t - \wt \Psi_t \no^2 + C \vert \mathcal E_{H^0}( \Psi_0) - \mathcal E_{h^{0,\varphi_0}}(\varphi_0) \vert +  C^{\varphi_t}  \no \cfs{f}{0}{\text{odd}}  \Psi_0 \no^2 \\
 & \hspace{5cm} +  e^{C^{\varphi_t}} \Big( N \no \q{1}{0}  \q{2}{0} \Psi_0 \no^2 + \no \q{1}{0} \Psi_0 \no^2  +\frac{1}{N} \Big).
\end{align*}
\end{lemma}
\begin{remark} Similar regularity properties of $\q{1}{t}\Psi_t$ have been used in the derivation of the Gross-Pitaevskii equation \cite{pickl:2010gp_ext} or the Hartree equation for very singular potentials \cite{picklknowles:2010}, both in the sense of reduced density matrices. The main difference here is that the explicit error on the r.h.s.\ is of order $1/N$ instead of $1/\sqrt{N}$ as, e.g., in \cite[Lemma~4.6]{picklknowles:2010}. This improvement comes for the price of having the additional terms $\no \Psi_t - \wt \Psi_t \no^2 $ and $\no \cfs{f}{0}{\text{odd}}  \Psi_0 \no^2$. The reason for the better explicit error is the result from Lemma \ref{LEMMA:ODD:EVEN:PART:TIME:EVOLUTION}, i.e., the fact that $\wt H^t$ does not couple the odd and even sectors of the wave function. Correlations between odd and even parts, in particular terms like $\lsp \Psi_t, \p{1}{t} \Delta_1 \q{1}{t} \Psi_t \rsp$, lead to the convergence rate $1/\sqrt{N}$ in \cite[Lemma~4.6]{picklknowles:2010}. That such correlations are not created during time evolution is, of course, not a priori known for the full Hamiltonian $H^t$.
\end{remark}

From now on, let $\varphi_t$ be the solution to the Hartree equation \eqref{TIME:DEP:HARTREE:EQUATION} with initial condition $\varphi_0\in L^2(\mathbb R^3)$, $\no \varphi_0\no=1$, let $\Psi_t = U_t \Psi_0$ and $\wt \Psi_t = \widetilde U_t \Psi_0$ with $\Psi_0\in L^2_s(\mathbb R^{3N})$, $\no \Psi_0 \no=1$, and let the initial wave functions $\varphi_0$, $\Psi_0$ satisfy\\
\\
\noindent \textbf{Assumptions A.2.}\label{A.2} \label{ASSUMPTIONS:A_2}
\begin{enumerate}
\item $\varphi_0 \in H^2(\mathbb R^2)$, $\Psi_0\in H^1(\mathbb R^{3N})$, and $\vert {\mathcal E}_{H^0}(\Psi_0) - \mathcal E_{h^{0,\varphi_0}}(\varphi_0) \vert \le C N^{-1}$ \label{A.2.1}
\item $ \lsp \Psi_0 , \Big(\prod_{i=1}^n \q{i}{0} \Big) \Psi_0 \rsp \le C N^{-n}$ for $n=1,2,3$, \label{A.2.2}
\item $\no \cfs{f}{0}{\text{odd}}  \Psi_0 \no \le  CN^{-\frac{1}{2}}$.\label{A.2.3}
\end{enumerate}
\begin{remark} Note that instead of \hyperref[A.2.3]{A.2.3}, one could equivalently assume that the even part of the wave function is initially small, i.e., $\no \cfs{f}{0}{\text{even}}  \Psi_0 \no \le CN^{-\frac{1}{2}}$. This would lead to the exact same results with all proofs being completely analogous (therefore we restrict ourselves to \hyperref[A.2.3]{A.2.3}).
\end{remark}

\begin{theorem}
\label{THEOREM:MAIN:THEOREM:NORM:APPROXIMATION} Let $\varphi_0, \Psi_0$ satisfy Assumptions \hyperref[A.2]{A.2}.\ Then there exists a time-dependent constant $C^{\varphi_t}>0$ such that for all $t\ge 0$,
\begin{align}
\no \Psi_t - \widetilde \Psi_t \no  \le \frac{e^{ C^{\varphi_t} }}{ \sqrt{N}}.\label{NORM:BOUND:MAIN:THEOREM}
\end{align}
\end{theorem}
\begin{remark} It can be verified along the lines of the proof that the constant that appears on the r.h.s.\ is given in terms of a polynomial $P: \mathbb R \to \mathbb R^+$ s.t.\ $C^{\varphi_t} = \int_0^t P(\no \Delta \varphi_s \no) ds$.
\end{remark}

\begin{remark}\label{IDEA:OF:PROOF} To obtain convergence in \eqref{NORM:BOUND:MAIN:THEOREM} for bounded potentials is very straightforward. The idea is to apply Duhamel's formula leading to
\begin{align*}
\no \Psi_t - \widetilde \Psi_t \no^2  = & \Big\vert 2 \int_0^t   \im \lsp \Psi_s, ( H^s - \wt H^s ) \wt \Psi_s \rsp ds  \Big\vert,
\end{align*}
and then use that $\lambda_N=1 / (N-1)$, together with the fact that each term in  $H^s - \wt H^s$ contains three or four $q^s$'s, cf.\ \eqref{DIFFERENCE:HAMILTONIAN}. For a bounded potential, convergence then follows immediately from Lemma \ref{A:PRIORI:ESTIMATE:PSI}.
\end{remark}
\begin{remark}\label{remark:ground:state}
One example of a wave function that satisfies Assumptions \hyperref[A.2]{A.2} is the complete product state $\Psi_0=\varphi_0^{\otimes N}$, which is, however, very far from the ground state of an interacting system in the $L^2$-norm sense. As we already mentioned in the introduction, the picture behind the experimental realization of Bose-Einstein condensation is that one starts with the ground state of a trapped gas and removes the trap or disturbs the gas in another way, and then observes the nontrivial time evolution of the particles.\ It would be a nice completion of the overall argument if one could show that the ground state of a Hamiltonian $H^0$ of the form \eqref{HAMILTONIAN:H} satisfies all three conditions from Theorem \ref{THEOREM:MAIN:THEOREM:NORM:APPROXIMATION}. In \cite{Lewin:2015b}, e.g., it was shown that the next-to-leading order corrections to the ground state wave function $\Psi^{(0)}$ can be described in terms of a family of correlation functions $(\chi^{(k)})_{k\ge 0}$ in the same sense as was explained for the time-dependent problem in \eqref{TIME:DEP:FLUCTUATIONS}.\ It was shown that $(\chi^{(k)})_{k\ge 0}$, interpreted as an element of the bosonic Fock space, is quasi free.\ The quasi free property is similar but not equivalent to conditions \hyperref[A.2.2]{A.2.2} and \hyperref[A.2.3]{A.2.3}, and the general results from \cite{Lewin:2015b} are not sufficient to show that \hyperref[A.2.2]{A.2.2} and \hyperref[A.2.3]{A.2.3} hold quantitatively for the situation that we have in mind. In a forthcoming work, we answer this question affirmatively in a more restricted setting, namely for the homogeneous gas on the torus, i.e., for $\mathbb R^{3N}$ replaced by $\mathbb T^{3N}$ with $\vert \mathbb T\vert <\infty$, and $W^0=0$. We expect the ground state to satisfy Assumptions \hyperref[A.2]{A.2} also in a more general setting, in particular for the initially trapped gas, for which it is technically much more difficult to analyze the ground state properties and to show that the static (or initial) approximation is as good as the one of the dynamics.
\end{remark}
\subsection{Trace norm convergence}

Using the previous theorem together with Lemmas \ref{A:PRIORI:ESTIMATE:PSI} and \ref{LEMMA:ODD:EVEN:PART:TIME:EVOLUTION}, one can show that $\gamma_{\Psi_t}^{(1)} \to \vert \varphi_t \rangle \langle \varphi_t \vert$ with rate $1/N$. We emphasize that by means of controlling only the relative number of particles in the condensate instead of the full $L^2$-approximation of $\Psi_t$, trace norm convergence of the reduced density can only be shown with error $\propto 1 / \sqrt N$,  \vspace{-0.05cm}cf.\ Remark \ref{REMARK:CONV:RATE}. It follows also straightforwardly, using in addition Lemma \ref{LEMMA:ENERGY:ESTIMATE:PSI:TILDE}, that $\gamma_{\Psi_t}^{(1)}$ is close to $\vert \varphi_t \rangle \langle \varphi_t \vert$ in terms of the energy trace distance.

\begin{theorem}
\label{CONVERGENCE:OF:DENSITY} Let $\varphi_0, \Psi_0$ satisfy Assumptions \hyperref[A.2]{A.2}.\ Then there exists a time-dependent constant $C^{\varphi_t}$ such that for all $t\ge 0$,
\end{theorem}
\vspace{-0.9cm}
\begin{align}\label{CONVERGENCE:OF:DENSITY:a}
\text{Tr} \big\vert \gamma^{(1)}_{\Psi_t} - \vert \varphi_t \rangle \langle \varphi_t \vert \big\vert & \le \frac{e^{C^{\varphi_t}}}{N},\\
\text{Tr} \big\vert \sqrt{1-\Delta} \big( \gamma^{(1)}_{\Psi_t} - \vert \varphi_t \rangle \langle \varphi_t \vert \big) \sqrt{1-\Delta} \big\vert & \le \frac{e^{C^{\varphi_t}}}{\sqrt N}.
\end{align}
\begin{remark}For a proof of the first statement, we actually require less regularity of $\varphi_0$ than stated in Assumptions \hyperref[A.2]{A.2}. To this end note that one finds a similar norm approximation as in Theorem \ref{THEOREM:MAIN:THEOREM:NORM:APPROXIMATION}, using an effective Hamiltonian defined by
\begin{align} \label{H:TILDE:4:Q}
 \widetilde H^t + \lambda_N \sum_{1\le i\neq j\le N} \q{i}{t} \q{j}{t} \big( v_{ij} - \bar v_i^{\varphi_t} \big) \q{i}{t} \q{j}{t}.
\end{align}
In this case, Lemma \ref{LEMMA:ENERGY:ESTIMATE:PSI:TILDE} together with Assumption \hyperref[A.2.1]{A.2.1} are not needed to derive a result analogous to \eqref{NORM:BOUND:MAIN:THEOREM}, and therefore it is sufficient to assume $\varphi_0 \in H^1(\mathbb R^3)$.\ We omit further details since the indicated argument can be readily verified along the steps of the proof of Theorem~\ref{THEOREM:MAIN:THEOREM:NORM:APPROXIMATION} when $\widetilde H^t$ is replaced by \eqref{H:TILDE:4:Q}.
\end{remark}
\begin{remark}\label{REMARK:CONV:RATE} The vanishing of the r.h.s.\ in \eqref{CONVERGENCE:OF:DENSITY:a} for large $N$ is a long known result, see references below \eqref{PROPAGATION:OF:CHAOS}. The optimal rate, to our knowledge, has only been derived for initial conditions equal to the full Hartree product \cite{ErdosSchlein:2008,ChenLee:2010,Chen:2011,Lee:2013}.\ Theorem \ref{CONVERGENCE:OF:DENSITY} holds for more general wave functions, and in particular, we expect it again to hold for the ground state of a trapped, interacting system (recall Remark \ref{remark:ground:state}). That it is not possible to improve the error further, can be inferred from \eqref{CONVERGENCE:ALPHA:TRACE}. The l.h.s.\ converging faster than $1/N$ would imply the wave function $\Psi_t$ to be close to the state $\varphi_t^{\otimes N}$ in $L^2$ sense.\ The latter is known to be false for interacting systems, as can be inferred, e.g., from \eqref{NORM:BOUND:MAIN:THEOREM}.
\end{remark}
\begin{remark} Following the argument from the proof of \eqref{CONVERGENCE:OF:DENSITY:a}, one can show as well that for any fixed integer $k$, the $k$-particle reduced density, $\gamma_{\Psi}^{(k)} : L^2(\mathbb R^{3k}) \to L^2(\mathbb R^{3k})$, defined by its kernel
\begin{align*}
\gamma_{\Psi}^{(k)}(x_1,...,x_k,y_1,...,y_k) = \int  \Psi_t(x_1,...,x_k,x_{k+1}...x_N) \overline{\Psi_t (y_1,...,y_k,x_{k+1}...x_N}) dx_{k+1}...dx_{N},
\end{align*}
converges to the $k$-fold product of the Hartree density, i.e.,
\begin{align*}
\text{Tr} \big\vert \gamma^{(k)}_{\Psi_t} - \vert \varphi_t \rangle \langle \varphi_t \vert^{\otimes k} \big\vert & \le \frac{e^{C_k^{\varphi_t}}}{N}.
\end{align*} 
\end{remark}

\subsection{Bogoliubov corrections on Fock space\label{BOG:DE:GENNES:CORRECTIONS}}

We define the set of correlation functions $ (\wt \chi^{(k)} )_{k=0}^N$, $ \wt \chi^{(k)} \in \q{1}{t}\otimes...\otimes \q{k}{t}L^{2}_s(\mathbb R^{3k})  \equiv \mathcal H^{(k),t}_s$, by
\begin{align}
\wt \chi^{(k)}_t(x_1,...,x_k)= \sqrt{\frac{N!}{k!(N-k)!}} \left(\prod_{i=1}^k\q{i}{t}\right) \int \left(\prod_{i=k+1}^N \overline{\varphi_t(x_i)}\right) \widetilde \Psi_t(x_1,...,x_N) \,dx_{k+1}\ldots dx_N.\label{TIDLE:CHI:K:DEFINITION}
\end{align}
By means of the partition $1 = \sum_{k=0}^N ( \q{1}{t}...\q{k}{t} \p{k+1}{t}...\p{N}{t})_{sym}$, c.f.\ Definition \ref{P:FAMILY:OPERATORS}, one can show that the following time-dependent decomposition of $\wt \Psi_t$, in terms of $\varphi_t$ and the correlation functions $\wt \chi_t^{(k)}$, holds as an identity at all times,
\begin{align}
\widetilde \Psi_t  = \sum_{k=0}^N \varphi_t^{\otimes N-k} \otimes_s  \wt \chi_t^{(k)} . \label{DECOMPOSITION:PSI:TILDE}
\end{align}
Here, $\otimes_s$ stands for the normalized symmetric tensor product between $\psi^{(l)} \in L^2(\mathbb R^{3l})$ and $\psi^{(k)} \in L^2(\mathbb R^{3k})$ defined by
\begin{align}\label{SYMMETRIC:PRODUCT}
 \psi^{(l)}  \otimes_s \psi^{(k)} = \frac{1}{\sqrt{k!l!(k+l)!}} \sum_{\sigma \in P_{k+l}} \psi^{(l)}(x_{\sigma(1)},...,x_{\sigma(l)})  \psi^{(k)}(x_{\sigma(l+1)},...,x_{\sigma(l+k)}).
\end{align}
It follows from \eqref{TIDLE:CHI:K:DEFINITION} that the $\wt \chi^{(k)}_t$ are orthogonal to $\varphi_t$ at all times, as well as that \vspace{-0.19cm} $\no \wt \chi^{(k)}_t\no_{\mathcal H_s^{(k),t}}^2$ equals the probability of finding exactly $k$ particles in $\wt \Psi_t$ which are not in the condensate wave function. The idea of decomposing an $N$-particle wave function according to \eqref{DECOMPOSITION:PSI:TILDE} and to study the thereby defined $k$-particle correlation functions was introduced in \cite{Lewin:2015b} where it was used to fully characterize the low energy spectrum (eigenvalues and eigenvectors) of the Bose gas in the mean field limit. The idea was then used to study the time evolution in \cite{Lewin:2015a} for the mean field limit and similarly in \cite{nam:2015,nam:2016} for the NLS scaling.\\

We next introduce the coupled hierarchy of Bogoliubov equations determining the time evolution of an infinite set of correlation functions which we denote by $(\chi^{(k)}_t)_{k\ge 0}$, $\chi^{(k)}_t(x_1,...,x_k) \in \mathcal H_s^{(k),t}$,
\begin{align}\label{Bog_hierarchy}
i\partial_t   \chi^{(0)}_t = &  \frac{1}{\sqrt{2}}\int \int \overline{ K^{(2),t}(x,y)}  \chi^{(2)}_t(x,y)  dxdy   , \nonumber\\
i\partial_t    \chi^{(1)}_t(x_1) = &  \Big( h^{\varphi_t}_1 + K^{(1),t}_1 \Big) \chi^{(1)}_t(x_1)  + \frac{\sqrt{6}}{2}\int \int    \overline{ K^{(2),t}(x,y) }  \chi_t^{(3)}(x_1,x,y) dxdy ,\nonumber\\
  i\partial_t \chi^{(k)}_t(x_1,...,x_k) = & \sum_{i=1}^k \Big( h^{\varphi_t}_i + K^{(1),t}_i \Big) \chi^{(k)}_t(x_1,...,x_k)\nonumber\\
  &  +  \frac{1}{2\sqrt{k(k-1)}} \sum_{1\le i<j\le k} K^{(2),t}(x_i,x_j) \chi^{(k-2)}_t(x_1,...,x_k \backslash x_i \backslash x_j)\nonumber\\
  & + \frac{\sqrt{(k+1)(k+2)}}{2} \int \int   \overline{K^{(2),t}(x,y)} \chi_t^{(k+2)}(x_1,...,x_k,x,y) dxdy
\end{align}
for all $k\ge 2$. Here, $K^{(1),t} : L^2(\mathbb R^3) \to L^2(\mathbb R^3)$ is given by $ K^{(1),t} = q^t \wt K^{(1),t} q^t$ where $\wt K^{(1),t} : L^2(\mathbb R^3) \to L^2(\mathbb R^3)$ is defined via its integral kernel $\wt K^{(1),t}(x,y) = \overline{\varphi_t(y)} v(x-y) \varphi_t(x)$. Further, $K^{(2),t}: L^2(\mathbb R^3)\otimes L^2(\mathbb R^3) \to q^t L^2(\mathbb R^3)\otimes q^t L^2(\mathbb R^3)$ with  $K^{(2),t} = q^t\otimes q^t \wt K^{(2),t}$ where  $\wt K^{(2),t}(x,y) = v(x-y) \varphi_t(x) \varphi_t(y)$.
\begin{remark} [Equivalence to Bogoliubov theory on Fock space] Interpreting $(\chi^{(k)}_t)_{k\ge 0} = \mathbf{\chi}_t$ as an element of the time-dependent Fock space over the excitations around the Hartree state $\varphi_t$, i.e., $\mathcal F_s ( \mathcal H^{(1),t} ) = \bigoplus_{n=0}^\infty  \mathcal H_s^{(n),t},$ the above hierarchy \eqref{Bog_hierarchy} is equivalent to the Schr\"odinger equation
\begin{align}\label{BOG:SCHROEDINGER:EQ}
i \partial_t \mathbf{ \chi}_t = \mathbb H_{Bog}^t \mathbf{ \chi }_t
\end{align}
on $\mathcal F_s ( \mathcal H^{(1),t} )$. The quadratic Hamiltonian $\mathbb H_{Bog}^t$ is nonparticle conserving and given by 
\begin{align*}
\mathbb H_{Bog}^t =  & \int  a^*_x \big( h^{t,\varphi_t}_x + K^{(1),t}_x \big) a_x \ dx + \frac{1}{2}\int \int \Big[ K^{(2),t}(x,y) a^*_x a^*_y   +  \overline{ K^{(2),t}(x,y)} a_x a_y \Big]  dxdy,
\end{align*}
where the operator-valued distributions $a^*_x$, $a_x$ are defined as
\begin{align*}
(a^*_x \chi)^{(k)}(x_1,...,x_k) = & \frac{1}{\sqrt k} \sum_{i=1}^k \delta(x_i-x) \chi^{(k-1)}(x_1,...x_k \backslash x_i), \\
(a_x \chi)^{(k)}(x_1,...,x_k) = & \sqrt{k+1} \int  \chi^{(k+1)}(x_1,...x_k,x) dx .
\end{align*}
The decomposition of $\wt \Psi_t$ in \eqref{DECOMPOSITION:PSI:TILDE} corresponds then to a partial isometry given by
\begin{align*}
\mathbb U^{\varphi_t}: L^2_s(\mathbb R^{3N}) \to  \mathcal F_s ( \mathcal H^{(1),t} ) , \hspace{1cm} \Psi\mapsto \mathbb U^{\varphi_t}\Psi = \chi_t^{(0)} \oplus \chi_t^{(1)} \oplus ... \oplus \chi_t^{(N)} \oplus 0 \oplus 0 \oplus ...,
\end{align*}
where the $\chi^{(k)} \in \mathcal H_s^{(k),t}$ are defined as in \eqref{TIDLE:CHI:K:DEFINITION} with $\wt \Psi_t$ replaced by $\Psi$.  For more details, we refer the reader to \cite{Lewin:2015a} and \cite{nam:2015}.
\end{remark}
\begin{remark}[Equivalence to Bogoliubov theory for density matrices] Yet another way to understand the Bogoliubov hierarchy was considered in \cite{nam:2015}, motivated by ideas from \cite{GrillakisMachedon:2013}. If one defines the density matrices $\gamma_t : \mathcal H^{(1),t} \to \mathcal H^{(1),t}$ and $\alpha_t : \overline{ \mathcal  H^{(1),t} } \to \mathcal H^{(1),t}$ by 
$$\langle f, \gamma_t g\rangle = \lsp \chi_t , a^*(g) a(f) \chi_t \rsp_{\mathcal F_s(\mathcal H^{(1),t})}, \hspace{1cm} \langle f, \alpha_t \overline{g} \rangle = \lsp \chi_t, a(g) a(f) \chi_t \rsp_{\mathcal F_s(\mathcal H^{(1),t})}, $$ 
the hierarchy \eqref{Bog_hierarchy} is equivalent to the pair of coupled equations for $\gamma_t$ and $\alpha_t$,
\begin{align*}
i \partial_t \gamma_t & =  \big( h^{t,\varphi_t} + K^{(1),t} \big) \gamma_t - \gamma_t \big( h^{t,\varphi_t} + K^{(1),t} \big) + K^{(2),t} \alpha_t - \alpha_t^{*} K^{(2),t*}, \\
 i \partial_t \alpha_t & =  \big( h^{t,\varphi_t} + K^{(1),t} \big) \alpha_t + \alpha_t \big( h^{t,\varphi_t} + K^{(1),t} \big)^{\text T} + K^{(2),t} +  K^{(2),t} \gamma_t^{\text{T}} + \gamma_t K^{(2),t}.
\end{align*}
\end{remark}
\begin{remark} Well-posedness of the Bogliubov hierarchy \eqref{Bog_hierarchy}, or equivalently of \eqref{BOG:SCHROEDINGER:EQ} has been shown in \cite[Section 4.3]{Lewin:2015a}. The main difficulty is the time-dependence of $\mathbb H_{Bog}^t$, \vspace{-0.1cm} and one essential ingredient is to show that $K^{(2)}_t$ is a Hilbert-Schmidt operator. This corresponds to the physical fact that only a finite number of correlations (particles in $\mathcal F_s(\mathcal H^{(1),t})$) is created during time evolution.
\end{remark}

Our last goal is to show that the corrections to the Hartree product in $\Psi_t$ are effectively described by the solutions of the Bogoliubov hierarchy, i.e.,
\begin{align}
\lim_{N\to \infty} \lno \Psi_t - \sum_{k=0}^N \varphi_t^{\otimes N-k} \otimes_s \chi_t^{(k)} \rno = 0. \label{GOAL:SECTION:BOG:CORRECTIONS}
\end{align}
To this end, it remains to show that $\lim_{N\to\infty} \wt \chi_t^{(k)} = \chi_t^{(k)}$.
\begin{lemma}\label{LEMMA:CHI:TILDE:CHI:BOG} Let $\varphi_0, \Psi_0$ satisfy Assumptions \hyperref[A.2]{A.2}.\ If $(\widetilde \chi_t^{(k)})_{k=0}^N$ is given by \eqref{TIDLE:CHI:K:DEFINITION}, and $(\chi_t^{(k)})_{k\ge 0}$ solves the Bogoliubov hierarchy \eqref{Bog_hierarchy} with initial condition $(\chi_0^{(k)} = \wt \chi^{(k)}_0)_{k=0}^N$ and $(\chi_0^{(k)}=0) _{k \ge N+1} $, then there exists a time-dependent constant $C^{\varphi_t}>0$ such that for all $t\ge 0$,
\begin{align*}
\sum_{k=0}^N \no \wt \chi_t^{(k)} - \chi_t^{(k)} \no^2_{\mathcal H^{(k),t}_s} \le \frac{e^{C^{\varphi_t}}}{N}.
\end{align*}
\end{lemma}
\begin{remark}
We emphasize that all elements of the tuple $(\wt \chi_t^{(k)})_{k=0}^N$ depend explicitly on $N$ whereas the sequence $(\chi_t^{(k)})_{k\ge 0}$ is $N$-independent.
\end{remark}
A quantitative version of \eqref{GOAL:SECTION:BOG:CORRECTIONS} follows as a simple corollary of Theorem \ref{THEOREM:MAIN:THEOREM:NORM:APPROXIMATION} and the previous lemma.
\begin{theorem}\label{THM:BOG:HIERARCHY}
Let $\varphi_0, \Psi_0$ satisfy Assumptions \hyperref[A.2]{A.2}. If $(  \chi_t^{(k)})_{k\ge 0}$ solves the Bogoliubov hierarchy \eqref{Bog_hierarchy} with initial condition
\begin{align*}
\chi^{(k)}_0(x_1,...,x_k)= \sqrt{\frac{N!}{k!(N-k)!}} \left(\prod_{i=1}^k\q{i}{0}\right) \int \left(\prod_{i=k+1}^N \overline{\varphi_0(x_i)}\right) \Psi_0(x_1,...,x_N) \,dx_{k+1}\ldots dx_N
\end{align*}
for $0\le k \le N$, and $(\chi_0^{(k)}=0)_{k\ge N+1}$, then there exists a time-dependent constant $C^{\varphi_t}>0$ such that for all $t\ge 0$,
\begin{align*}
\lno \Psi_t - \sum_{k=0}^N \varphi_t^{\otimes N-k} \otimes_s \chi^{(k)}_t \rno \le \frac{e^{C^{\varphi_t}}}{\sqrt{N}}.
\end{align*}
\end{theorem}

\section{Proofs\label{sec:Proofs}}

We first state a technical lemma from which the proofs of the theorems then follow easily. We defer its proof to Section \ref{PROOFS:OF:LEMMATA}.\ It can essentially be read as estimates for terms like $\vert \lsp \Phi, \big[ q_1^{\varphi} , A_1 \big] \wt \Phi \rsp \vert$,
where $A_1$ is a one-particle operator and $\Phi, \wt \Phi$ symmetric wave functions. To have control of such terms is important in order to use each of the $q^t$'s that are available in the terms that need to be estimated.

\begin{lemma}\label{Q:COMMUTATION} Let $\varphi \in L^2(\mathbb R^3)$, $\Phi, \wt \Phi \in L_s^2(\mathbb R^{3N})$, and $p^{\varphi}_i =  \vert \varphi(x_i) \rangle \langle \varphi(x_i) \vert $, $q_i^{\varphi} = 1- p^{\varphi}_i$ as in Definition \eqref{DEF:PROJECTORS} and $\widehat f_{\text{odd}}^\varphi$, $\widehat f_{\text{even}}^\varphi$ as in \eqref{DEF:F:ODD} and \eqref{DEF:F:EVEN}. 

\begin{enumerate}
\item Let $A_1$ be an operator on $L^2(\mathbb R^3)$ with $\no A_1 p^{\varphi}_1 \no_{op} < C $. Then
\begin{align}
\big\vert \lsp \Phi, q^{\varphi}_1 A_1p^{\varphi}_1 \Phi \rsp \big\vert + \big\vert \lsp \Phi, p^{\varphi}_1 A_1q^{\varphi}_1 \Phi \rsp \big\vert  \le C \no \Phi- \wt \Phi \no^2  + C \no \widehat f_{\text{odd}}^\varphi \wt \Phi \no^2 + C^{\varphi} \Big( \no q^{\varphi}_1  \Phi \no^2 + \frac{1}{N} \Big). \label{QAP:COMMUTATION}
\end{align} 
\item Let $v=v(x)$ satisfy $v^2 \le C(1-\Delta)$. Then
\begin{align} 
\big\vert \lsp \wt \Phi,  q^{\varphi}_1 q^{\varphi}_2 v(x_1-x_2) q^{\varphi}_1 q^{\varphi}_2 \Phi \rsp \big\vert  \le \frac{C N  \no  q^{\varphi}_1 q^{\varphi}_2 q^{\varphi}_3 \wt \Phi \no^2}{2} +  \frac{\no \nabla_1 q_1^{\varphi} \Phi \no^2}{2N} + \frac{\no q_1^{\varphi} \Phi \no^2}{2N}. \label{QQVQQ:COMMUTATION}
\end{align} 
\item Let $A_{12} = A_{21}$ be an operator on $L^2(\mathbb R^3) \otimes L^2(\mathbb R^3)$ with $\no A_{12}p^{\varphi}_{2} \no_{op}<C$. Then
\begin{align}
& \big\vert \lsp  \Phi,  \big( q^{\varphi}_1 q^{\varphi}_2 A_{12} q^{\varphi}_1 p^{\varphi}_2 + \text{h.c.} \big) \wt \Phi \rsp \big\vert \nonumber\\
& \hspace{2cm} \le \frac{C \no \Phi - \wt \Phi \no^2}{N} + \frac{C \no \widehat f_{\text{odd}}^\varphi \wt \Phi \no^2}{N} + C^{\varphi_t} N \Big( \no  q^{\varphi}_1 q^{\varphi}_2 q^{\varphi}_3 \Phi \no^2  + \no  q^{\varphi}_1 q^{\varphi}_2 q^{\varphi}_3 \wt \Phi \no^2  \Big). \label{QQAQP:COMMUTATION} 
\end{align}
\end{enumerate}
\end{lemma}

\subsection{Proofs of Theorems \ref{THEOREM:MAIN:THEOREM:NORM:APPROXIMATION}, \ref{CONVERGENCE:OF:DENSITY} and \ref{THM:BOG:HIERARCHY} \label{PROOF:OF:THEOREMS}}

\begin{proof}[Proof of Theorem \ref{THEOREM:MAIN:THEOREM:NORM:APPROXIMATION}] Our goal is to estimate the time derivative of $\no \Psi_t - \widetilde \Psi_t \no^2$ in terms of itself and a small error, i.e., we are after a bound of the type
\begin{align*}
\partial_t \no \Psi_t - \widetilde \Psi_t \no^2 \le C \no \Psi_t - \widetilde \Psi_t \no^2 + \frac{C^{\varphi_t}}{N}.
\end{align*}
Then, Gr\"onwall's Lemma, together with $\Psi_{t=0} = \wt \Psi_{t=0} = \Psi_0$, implies \eqref{NORM:BOUND:MAIN:THEOREM} (note that by a standard density argument the following calculations hold for all $\Psi_0\in L^2_s(\mathbb{R}^{3N})$).

Using that $\wt H^t$ is self-adjoint together with \eqref{DIFFERENCE:HAMILTONIAN} and the symmetry of $\Psi_t$ and $\widetilde \Psi_t$ , we find
\begin{align}
\partial_t \no \Psi_t - \wt \Psi_t \no^2 
  = & 2 \im  \lsp \Psi_t- \wt \Psi_t,  H^t \Psi_t -  \wt H^t  \wt \Psi_t \rsp \nonumber \\
  = & 2 \im  \lsp \Psi_t- \wt \Psi_t,  \big( H^t - \wt H^t \big)  \Psi_t \rsp - 2\im \lsp \Psi_t-\wt \Psi_t,  \wt H^t \big( \Psi_t-\wt \Psi_t \big) \rsp \nonumber \\   
  = &  - 2  \im \lsp \wt \Psi_t, (H^t-\wt H^t)  \Psi_t \rsp  \nonumber\\
  = & -4 N \im \lsp \wt \Psi_t,  \big( \q{1}{t} \q{2}{t} (v_{12} - \bar v_1^{\varphi_t} ) \q{1}{t} \p{2}{t} + \text{h.c.} \big)  \Psi_t \rsp \label{GROENWALL:ESTIMATE:NORM:DIFFERENCE:1} \\
  & - 4 N \im \lsp \wt \Psi_t,  \q{1}{t} \q{2}{t}  v_{12} ,  \q{1}{t} \q{2}{t} \Psi_t \rsp  + 4N \lsp \wt \Psi_t  ,  \q{1}{t} \q{2}{t} (\bar v_1^{\varphi_t} - \mu^{\varphi_t}) \q{1}{t} \q{2}{t}  \Psi_t \rsp.  \label{GROENWALL:ESTIMATE:NORM:DIFFERENCE:2}
\end{align}
For the first line, we use \eqref{BOUNDS:ON:V:PHI} from which it follows that $\no v_{12} \p{2}{t} \no_{op}\le C$, and then apply \eqref{QQAQP:COMMUTATION}. Then with Lemmas \ref{A:PRIORI:ESTIMATE:PSI} (for $n\le 3$) and \ref{LEMMA:ODD:EVEN:PART:TIME:EVOLUTION}, together with Assumptions \hyperref[A.2]{A.2}, we obtain
\begin{align*}
\vert \eqref{GROENWALL:ESTIMATE:NORM:DIFFERENCE:1} \vert \le & C \no \Psi_t - \wt \Psi_t \no^2  +  C \no  \cfs{f}{t}{\text{odd}} \wt \Psi_t \no^2  + C^{\varphi_t} N^2 \Big( \no \q{1}{t} \q{2}{t} \q{3}{t} \Psi_t \no^2  + \no \q{1}{t}\q{2}{t}\q{3}{t} \wt \Psi_t \no^2  \Big) \\
\le & C \no \Psi_t - \wt \Psi_t \no^2 + \frac{C^{\varphi_t} }{N}.
\end{align*}
For the first term in the second line, we proceed by means of \eqref{QQVQQ:COMMUTATION}, and then apply Lemmas \ref{A:PRIORI:ESTIMATE:PSI} and \ref{LEMMA:ENERGY:ESTIMATE:PSI:TILDE},
\begin{align*}
N \big\vert \lsp \wt \Psi_t,  \q{1}{t} \q{2}{t}  v_{12} ,  \q{1}{t} \q{2}{t} \Psi_t \rsp \big\vert \le  \frac{C N^2 \no   \q{1}{t} \q{2}{t}  \q{2}{t} \wt \Psi \no^2}{2} + \frac{C \no \nabla_1 \q{1}{t} \Psi_t \no^2}{2} \le C \no \Psi_t - \wt \Psi_t \no^2 + \frac{C^{\varphi_t} }{N}.
\end{align*}
The last term is small since $\no \bar v^{\varphi_t} \no_{op}\le C$, and thus via Lemma \ref{A:PRIORI:ESTIMATE:PSI} together with Assumption \hyperref[A.2.2]{A.2.2},
\begin{align*}
N \big\vert \lsp \wt \Psi_t  ,  \q{1}{t} \q{2}{t} (\bar v_1^{\varphi_t} - \mu^{\varphi_t}) \q{1}{t} \q{2}{t}  \Psi_t \rsp \big\vert \le \frac{C^{\varphi_t}}{N}.
\end{align*}
\end{proof}

\begin{proof}[Proof of Theorem \ref{CONVERGENCE:OF:DENSITY}]
We start from the fact that
\begin{align*}
\text{Tr} \big\vert \gamma^{(1)}_{\Psi_t} - \vert \varphi_t \rangle \langle \varphi_t \vert   \big\vert= \underset{\no A_1 \no \le 1}{\text{sup}} \big\vert \text{Tr} \big( A_1 \big( \gamma^{(1)}_{\Psi_t} -  \vert \varphi_t \rangle \langle \varphi_t \vert   \big) \big) \big\vert
\end{align*}
where the supremum is taken over compact operators $A_1$ acting on $L^2(\mathbb R^{3})$ with norm smaller or equal to one. Inserting the identity $1=\p{1}{t} + \q{1}{t}$ around $A_1$, we find
\begin{align}
\big\vert \text{Tr} \big( A_1 \big( \gamma^{(1)}_{ \Psi_t} - \vert \varphi_t \rangle \langle \varphi_t \vert   \big\vert \big) \big) \big\vert =   & \big\vert \lsp \Psi_t, A_1  \Psi_t\rsp - \lsp \varphi_t , A_1 \varphi_t  \rsp \big\vert \nonumber\\
\le    & \big\vert \lsp  \Psi_t, \p{1}{t} A_1 \p{1}{t}  \Psi_t\rsp - \lsp \varphi_t, A_1 \varphi_t  \rsp \big\vert +  \big\vert \lsp  \Psi_t, \q{1}{t} A_1 \q{1}{t} \Psi_t\rsp \big\vert \nonumber \\
& + \big\vert \lsp \Psi_t, \p{1}{t} A_1 \q{1}{t}  \Psi_t\rsp \big\vert + \big\vert \lsp \Psi_t, \q{1}{t} A_1 \p{1}{t}  \Psi_t\rsp \big\vert.\label{TRACE:NORM:CROSS:TERM}
\end{align}
The first term is small,
\begin{align*}
\big\vert \lsp  \Psi_t, \p{1}{t} A_1 \p{1}{t}  \Psi_t\rsp - \lsp \varphi_t, A_1 \varphi_t \rsp \big\vert = \vert \langle \varphi_t A_1 \varphi_t \rangle \vert \ \big\vert \lsp  \Psi_t, \big( \p{1}{t} - 1 \big)  \Psi_t\rsp \big\vert \le\frac{e^{C^{\varphi_t}}}{N},
\end{align*}
and so is the second term in the first line, both by means of Lemma \ref{A:PRIORI:ESTIMATE:PSI}. For the second line, we can use \eqref{QAP:COMMUTATION} in order to find
\begin{align*}
\vert \eqref{TRACE:NORM:CROSS:TERM} \vert \le C \no \Psi_t - \wt \Psi_t \no^2 + C^{\varphi_t} \Big( \no \q{1}{t}  \Psi_t \no^2 +   \no  \cfs{f}{t}{ \text{odd} }\wt  \Psi_t \no^2 + \frac{1}{N}\Big) \le \frac{e^{ {C^{\varphi_t}}}}{N},
\end{align*}
where the last step follows from Lemmas \ref{A:PRIORI:ESTIMATE:PSI}, \ref{LEMMA:ODD:EVEN:PART:TIME:EVOLUTION} and Assumptions \hyperref[A.2]{A.2}.\\

The second result follows from the estimate in Lemma \ref{LEMMA:ENERGY:ESTIMATE:PSI:TILDE}, together with Lemma \ref{A:PRIORI:ESTIMATE:PSI}. We start again from 
\begin{align*}
& \text{Tr} \big\vert \sqrt{1-\Delta} \big( \gamma^{(1)}_{\Psi_t} - \vert \varphi_t \rangle \langle \varphi_t \vert \big) \sqrt{1-\Delta} \big\vert = \underset{\no A_1 \no \le 1}{\text{sup}} \big\vert \text{Tr} \big( A_1  \sqrt{1-\Delta}  \big( \gamma^{(1)}_{\Psi_t} -  \vert \varphi_t \rangle \langle \varphi_t \vert   \big)  \sqrt{1-\Delta}  \big) \big\vert,
\end{align*}
where the supremum is taken over all compact operators with norm less or equal than one, and compute
\begin{align}
& \text{Tr} \big( A_1  \sqrt{1-\Delta}  \big( \gamma^{(1)}_{\Psi_t} -  \vert \varphi_t \rangle \langle \varphi_t \vert   \big)  \sqrt{1-\Delta}  \big) \nonumber\\
&= \lsp \Psi_t, \p{1}{t} \sqrt{1-\Delta_1}  A_1 \sqrt{1-\Delta_1} \p{1}{t} \Psi_t\rsp - \langle \varphi_t, \sqrt{1-\Delta}  A_1 \sqrt{1-\Delta}  \varphi_t\rangle \label{LINE:1:SOBOLEV:TRACE} \\
& + \lsp \Psi_t, \q{1}{t} \sqrt{1-\Delta_1}  A_1 \sqrt{1-\Delta_1} \p{1}{t} \Psi_t\rsp + \lsp \Psi_t, \p{1}{t} \sqrt{1-\Delta_1}  A_1 \sqrt{1-\Delta_1} \q{1}{t} \Psi_t\rsp \nonumber\\
& + \lsp \Psi_t, \q{1}{t} \sqrt{1-\Delta_1}  A_1 \sqrt{1-\Delta_1} \q{1}{t} \Psi_t\rsp.\nonumber
\end{align}
The first line,
\begin{align*}
\vert \eqref{LINE:1:SOBOLEV:TRACE} \vert = \langle \varphi_t, \sqrt{1-\Delta}  A_1 \sqrt{1-\Delta}  \varphi_t\rangle \no \q{1}{t} \Psi_t \no^2 \le \frac{e^{C^{\varphi_t}}}{N}
\end{align*}
since $\no \sqrt {1- \Delta_1} \p{1}{t} \no_{op}^2 \le \langle \varphi_t, (1-\Delta_1) \varphi_t\rangle$. The first term in the second line,
\begin{align*}
\big\vert \lsp \Psi_t, \q{1}{t} \sqrt{1-\Delta_1}  A_1 \sqrt{1-\Delta_1} \p{1}{t} \Psi_t\rsp \big\vert \le C \Big( \no \q{1}{t} \Psi_t \no + \no \nabla_1\q{1}{t} \Psi_t \no  \Big) \no \sqrt {1- \Delta_1} \p{1}{t} \no_{op} \le \frac{e^{C^{\varphi_t}}}{\sqrt N},
\end{align*}
since $\no \sqrt{1-\Delta_1} \q{1}{t} \Psi_t \no^2 \le \no \q{1}{t} \Psi_t \no^2 + \no \nabla_1 \q{1}{t} \Psi_t \no^2$. The second term in the second line is estimated in exactly the same way. For the term in the third line, we find
\begin{align*}
\lsp \Psi_t, \q{1}{t} \sqrt{1-\Delta_1}  A_1 \sqrt{1-\Delta_1} \q{1}{t} \Psi_t\rsp \le C \no \sqrt{1-\Delta_1} \q{1}{t} \Psi_t \no^2 \le \frac{e^{C^{\varphi_t}}}{N},
\end{align*}
which proves the estimate.
\end{proof}

\begin{proof}[Proof of Theorem \ref{THM:BOG:HIERARCHY}] Using the triangle inequality and Theorem \ref{THEOREM:MAIN:THEOREM:NORM:APPROXIMATION}, we know that 
\begin{align*}
\lno \Psi_t - \sum_{k=0}^N \varphi_t^{\otimes N-k} \otimes_s \chi^{(k)}_t \rno \le \frac{e^{C^{\varphi_t}}}{\sqrt{N}} + \no \widetilde \Psi_t- \sum_{k=0}^N \varphi_t^{\otimes N-k} \otimes_s \chi^{(k)}_t \no,
\end{align*}
where $\widetilde \Psi_t = \widetilde U_t\Psi_0$. Then, with Lemma \ref{LEMMA:CHI:TILDE:CHI:BOG} and $\no \varphi_t \no=1$, it follows that
\begin{align*}
\lno \widetilde \Psi_t- \sum_{k=0}^N \varphi_t^{\otimes N-k} \otimes_s \chi^{(k)}_t \lno^2 =   & \sum_{k=0}^N \no \widetilde\chi^{(k)}_t -  \chi^{(k)}_t \no_{\mathcal H^{t,k}_s}^2  \le \frac{e^{C^{\varphi_t}}}{N}.
\end{align*}
\end{proof}

\subsection{Preliminaries for proofs of the remaining lemmas}

We summarize some necessary definitions and preliminary assertions that are needed to prove the remaining lemmas. Readers familiar with the method that was introduced in \cite{pickl:2011method} can skip this section.

\begin{definition}\label{P:FAMILY:OPERATORS}
Define the time-dependent projectors $ ( P_{N,k}^{\varphi_t} ) _{k=0}^N $, $P_{N,k}^{\varphi_t}: L^2(\mathbb R^{3N}) \to L^2(\mathbb R^{3N})$ by 
\begin{align}\label{SYMMETRIC:TENSOR:PRODUCT}
P_{N,k}^{\varphi_t} = \Big( \prod_{i=1}^k\q{i}{t} \prod_{i=k+1}^N \p{i}{t} \Big)_{sym} =\sum_{\{ a\in \{0,1\}^N:\sum_{i}a_i=N \}} \prod_{i=1}^N (\q{i}{t} )^{a_i} (\p{i}{t})^{1-a_i}.
\end{align}
\end{definition}
\noindent The following properties hold:
\begin{enumerate}
\item $P_{N,k}^{\varphi_t}$ is an orthogonal projector,
\item $P_{N,k}^{\varphi_t} P_{N,l}^{\varphi_t} = \delta_{kl} P_{N,k}^{\varphi_t},$
\item $1 = \sum_{k= 0 }^N P_{N,k}^{\varphi_t},$
\item $[\p{l}{t}, P_{N,k}^{\varphi_t}]=0= [\q{l}{t} ,P_{N,k}^{\varphi_t} ]$.
\end{enumerate}

\begin{definition}
We call any function $f:\{0,1,...,N\} \to \mathbb R_0^+$ a weight function (or  simply weight) and define the linear combination of weighted projectors w.r.t.\ the weight $f$ by $\widehat f^{\varphi_t} \equiv \cf{f}{t}$,
\begin{align}
\cf{f}{t}: L^2(\mathbb R^{3N}) \to L^2(\mathbb R^{3N}), \hspace{1cm} \cf{f}{t} \Psi =\sum_{k=0}^N f(k)  P_{N,k}^{\varphi_t}\Psi. \label{COUNTING:OPERATOR}
\end{align}
For any integer $\vert d\vert \le N$, we define the shift operator $\tau_{d}$ by
\begin{align*}
 &\tau_d f :\{0,1,...,N\}\to \mathbb R_0^+, \hspace{1cm}  ( \tau_d f ) (k) = \begin{cases} 0 \ \ \ &\text{for} \ k+d < 0, \\
f(k+d)\ \  &\text{for} \ \ 0\le k+d\le N ,\\
0 \ \ \  &\text{for} \ N < k+d . \end{cases} 
\end{align*}
\end{definition}
\noindent It is straightforward to see that 
\begin{enumerate}
\item $[
\cf{f}{t} , \p{k}{t} ]=0 =[\cf{f}{t} , \q{k}{t} ]  , \hspace{0.5cm}[\cf{f}{t} ,P_{N,k}^{\varphi_t}]=0$,
\item $\cf{g}{t} \cf{f}{t} = \cf{f}{t} \cf{g}{t} = \cf{(fg)}{t} \label{PROPERTY:ALGEBRA:STRUCTURE}$ for any two weights $f,g$.
\end{enumerate}
We shall make frequent use of the weight functions
\begin{align}\label{IMPORTANT:WEIGHTS}
& m(k) = \frac{k}{N},\hspace{1cm}   n(k) = \sqrt{\frac{k}{N}}. 
\end{align}
They satisfy two important properties, namely
\begin{enumerate}
\item 
\begin{align}
\frac{1}{N}\sum_{k=1}^N \q{k}{t} =   \sum_{k=1}^N \frac{k}{N} P^{\varphi_t}_{N,k} = \cf{m}{t}, \label{PROPERTY:SUM:Q:EQUAL:M}
\end{align}
\item $$ \cf{\tau_d m}{t} = (\cf{\tau_d n}{t} )(\cf{\tau_d n}{t} ). $$
\end{enumerate}
We further introduce the weight functions
\begin{align}
\mu(k) = \begin{cases} 0  \ &\text{for}\  k=0 , \\
 \frac{N}{k}\  &\text{for} \ 1\le k \le N, \\
\end{cases} \hspace{1cm} \nu(k) = \begin{cases} 0  \ & \text{for}\  k=0  , \\
 \sqrt{\frac{N}{k}}\  & \text{for} \ 1\le k \le N, \\
\end{cases}
\end{align}
which satisfy 
\begin{enumerate}
\item [3.] 
\begin{align}
\cfs{m}{t}{} \cfs{\mu}{t}{} = 1 - P_{N,0}^{\varphi_t}, \hspace{1cm} \cfs{n}{t}{} \cfs{ \nu}{t}{} = 1  - P_{N,0}^{\varphi_t}, \label{PROPERTY:N:INVERSE:N}  
\end{align}
\item[4.] $$  \cf{\tau_d\mu}{t} = (\cf{\tau_d\nu}{t})(\cf{\tau_d\nu}{t}).$$
\end{enumerate}
Let us also note that the above definition is in agreement with \eqref{DEF:F:ODD} and \eqref{DEF:F:EVEN} if we set
$$ f_{\text{(even)}}(k) = \begin{cases} 1 \ &\text{for} \ k \ \text{even},\\
0 \ &\text{for} \ k \ \text{odd}, \end{cases} \hspace{1cm} f_{\text{(odd)}}(k) = 1-f_{\text{(even)}}(k).$$
It follows immediately that  $\cfs{f}{t}{(\text{odd})}  + \cfs{f}{t}{(\text{even})} =1$ and $\cfs{f}{t}{(\text{odd})}  \cfs{f}{t}{(\text{even})} =  \cfs{f}{t}{(\text{even})}  \cfs{f}{t}{(\text{odd})} = 0$.
\begin{lemma}\label{PROPOSITION:Q:K:M:N}
Let $\Psi, \wt \Psi \in L^2_s(\mathbb R^{3N})$ and $n+k\le N$. Then there exists a positive constant $C_{n,k}$ such that
\begin{align*}
\lsp \Psi , \Big(\prod_{i=1}^{k+n} \q{i}{t}\Big) \wt \Psi \rsp = & C_{n,k} \lsp \Psi  , (\cf{m}{t})^n \Big(\prod_{i=1}^k \q{i}{t}\Big) \wt \Psi_{\cdot} \rsp.
\end{align*}
\end{lemma}
\noindent The proof follows easily by using symmetry of $\Psi$ and $\wt \Psi$ together with \eqref{PROPERTY:SUM:Q:EQUAL:M}.
\begin{lemma}
\label{LEMMA:PULL:THROUGH:FORMULA} Let $P_{12}^{(0)} = \p{1}{t}\p{2}{t}$, $P_{12}^{(1)} = \p{1}{t} \q{2}{t}+\q{1}{t} \p{2}{t}$ and $P_{12}^{(2)} = \q{1}{t} \q{2}{t}$, and let $f$ be an arbitrary weight function, and $A_{12}$ any operator on  $L^2(\mathbb R^{3}) \otimes L^2(\mathbb R^{3})$. Then the following commutation rule (pull through formula) holds for $0\le i,j\le 2$:
\begin{align}
 P_{12}^{(i)} A_{12} P_{12}^{(j)}  \cf{f}{t} =  \cf{\tau_{j-i}f}{t}   P_{12}^{(i)} A_{12} P_{12}^{(j)}  \label{PULL:THROUGH:FORMULA}.
\end{align}
\end{lemma}
\noindent The proof is straightforward, see, e.g., \cite[Lemma~3.10]{picklknowles:2010}.
\begin{definition}
We define the so called counting functional w.r.t.\ $\varphi_t$ and with weight $f$ by
\begin{align*}
\lsp \cdot, \cf{f}{t} \cdot \rsp : \ L^2({\mathbb R^{3N}}) \to \mathbb R_0^+, \hspace{1cm}
\Psi \mapsto   \lsp \Psi , \cf{f}{t} \Psi \rsp.
\end{align*}
For any $\Psi \in L^2(\mathbb R^{3N})$, the mapping $t\to \lsp \Psi, \cf{f}{t} \Psi \rsp$ is time-differentiable with derivative
\begin{align}\label{TIME:DERIVATIVE:CF}
 \partial_t \lsp \Psi, \cf{f}{t} \Psi \rsp = - i \lsp \Psi, \Big[ \sum_{i=1}^N h^{\varphi_t}_i, \cf{f}{t} \Big] \Psi \rsp.
\end{align} 
\end{definition}

\begin{lemma}\label{TIME:DERIVATIVE:ALPHA} Let
\begin{align}
& Z^{\varphi_t}(x_1-x_2) = \frac{1}{2} \Big( v(x_1-x_2) - 
\big(v\ast \vert \varphi_t\vert^2\big)(x_1) - \big(v\ast \vert \varphi_t\vert^2\big)(x_2) \Big),
\end{align}
and, for any $\Psi\in L^2_s(\mathbb R^{3N})$, let
\end{lemma}
\vspace{-0.6cm}
\begin{align*}
\text{(I)}_{f,\Psi} &=   4 N \im \lsp {\Psi}, \p{1}{t}  \p{2}{t} Z^{\varphi_t}(x_1-x_2) \q{1}{t}  \p{2}{t}  \big( \cf{f}{t} - \cf{\tau_{-1}f}{t} \big)  {\Psi} \rsp , \nonumber\\
\text{(II)}_{f,\Psi} & =  2 N \im  \lsp \Psi, \p{1}{t}  \p{2}{t}  Z^{\varphi_t}(x_1-x_2)  \q{1}{t}  \q{2}{t} \big( \cf{f}{t} - \cf{\tau_{-2}f}{t} \big)  {\Psi} \rsp   , \nonumber\\
\text{(III)}_{f,\Psi} & = 4 N \im  \lsp {\Psi}, \q{1}{t}  \p{2}{t}  Z^{\varphi_t}(x_1-x_2) \q{1}{t} \q{2}{t} \big( \cf{f}{t} - \cf{\tau_{-1} f}{t}{} \big)  {\Psi} \rsp, \nonumber
\end{align*}
\textit{with $f$ being any weight function. Let further $\Psi_0\in L^2(\mathbb R^{3N})$. Then,}
\begin{enumerate}
\item[\textit{1.}] \textit{for $\Psi_t=U_t \Psi_0$, we have}
\end{enumerate}
\begin{align}
\partial_t \lsp  \Psi_t , \cf{f}{t} \Psi_t \rsp  = \text{(I)}_{f,\Psi_t} +\text{(II)}_{f,\Psi_t} +\text{(III)}_{f,\Psi_t}.
\end{align}
\begin{enumerate}
\item[\textit{2.}] \textit{for $\widetilde \Psi_t= \widetilde U_t \Psi_0$, we have}
\end{enumerate}
\begin{align}
\partial_t \lsp  \widetilde \Psi_t , \cf{f}{t} \widetilde \Psi_t \rsp  =  ( \text{II})_{f,\wt \Psi_t}.
\end{align}
\begin{lemma} \label{UNIFORM:ROEMISCHE:1-3} Let $m(k)= \frac{k}{N}$ as in \eqref{IMPORTANT:WEIGHTS}. It holds that for any $\Psi \in L^2_s(\mathbb R^{3N})$
\end{lemma}
\vspace{-0.1cm}
\begin{align*}
(\text{I})_{m^n, \Psi} = 0, \hspace{1cm} \vert ( \text{II})_{m^n, \Psi} \vert + \vert ( \text{III})_{m^n,\Psi} \vert \le & \sum_{l=0}^{n} \frac{C^{\varphi_t} \lsp \Psi, (\cf{m}{t})^{l}\Psi \rsp}{N^{n-l}}.
\end{align*}
 
\subsection{Proofs of Lemmas \ref{A:PRIORI:ESTIMATE:PSI},   \ref{LEMMA:ODD:EVEN:PART:TIME:EVOLUTION}, \ref{LEMMA:ENERGY:ESTIMATE:PSI:TILDE}, \ref{LEMMA:CHI:TILDE:CHI:BOG}, \ref{Q:COMMUTATION}, \ref{TIME:DERIVATIVE:ALPHA} and \ref{UNIFORM:ROEMISCHE:1-3}
\label{PROOFS:OF:LEMMATA}}
 
We begin with
\begin{proof}[Proof of Lemma \ref{Q:COMMUTATION}]

\noindent We denote by $\widehat f^{\varphi}$ the counting functional defined as in \eqref{COUNTING:OPERATOR} with $\varphi_t$ replaced by $\varphi$.\\
\\
1. We recall that $\widehat n^{\varphi} \widehat \nu^{\varphi} q_1^{\varphi} \Phi = q_1^{\varphi} \Phi$, and use $\widehat f^{\varphi}_{\text{odd}} p_1^{\varphi}  A_{1} q_1^{\varphi}  \widehat f^{\varphi}_{\text{odd}} = 0 = \widehat f^{\varphi}_{\text{even}} p_1^{\varphi}  A_{1} q_1^{\varphi}  \widehat f^{\varphi}_{\text{even}} $ which follows from the pull through formula and $\widehat f^{\varphi}_{\text{even}} \widehat f^{\varphi}_{\text{odd}} = 0$. Then,
\begin{align*}
& \big\vert  \lsp  \Phi, p_1^{\varphi} A_{1}  q_1^{\varphi}  \Phi \rsp \big\vert  \\
& =  \vert  \lsp \widehat f^{\varphi}_{\text{odd}} \Phi,   p_1^{\varphi}  A_1 q_1^{\varphi} \widehat f^{\varphi}_{\text{even}}  \Phi \rsp +  \lsp \widehat f^{\varphi}_{\text{even}}  \Phi,   p_1^{\varphi}  A_1    q_1^{\varphi} \widehat f^{\varphi}_{\text{odd}}  \Phi \rsp 
\big\vert \nonumber\\
& =  \big\vert  \lsp \widehat f^{\varphi}_{\text{odd}}  \Phi, p_1^{\varphi}  A_1 q_1^{\varphi} \widehat f^{\varphi}_{\text{even}}  \Phi \rsp +  \lsp  \widehat{\tau_{-1} n}^\varphi  \widehat f^{\varphi}_{\text{even}}  \Phi , p_1^{\varphi}  A_1   q_1^{\varphi} \widehat{ \nu}^\varphi  \widehat f^{\varphi}_{\text{odd}}   \Phi \rsp  \big\vert \nonumber\\
& \le  \no  \widehat f^{\varphi}_{\text{odd}}  \Phi \no \ \no  p_1^\varphi  A_1 \no_{op} \ \no  q_1^\varphi \widehat f^{\varphi}_{\text{even}} \Phi \no +  \no  \widehat{\tau_{-1} n}^\varphi  \widehat f^{\varphi}_{\text{even}} \Phi \no \ \no p_1^\varphi A_{1} \no_{op} \ \no  q^\varphi_1 \widehat{\nu}^\varphi \widehat f^\varphi_{\text{odd}} \Phi \no \\
& \le  \frac{\no  \widehat f_{\text{odd}}^\varphi  \Phi \no^2 }{2} + \frac{ \no  p_1^\varphi A_{1} \no^2}{2} \Big( \no  \q{1}{}  \Phi \no^2  + \frac{1}{N} \Big)\\
& \le   \no \widehat f_{\text{odd}}^\varphi \wt \Phi \no^2 +  \no  \Phi -  \wt \Phi \no^2  + C^{\varphi} \Big( \no  \q{1}{}  \Phi \no^2 + \frac{1}{N} \Big).
\end{align*}
2. Here, we recall $\widehat m^\varphi  \widehat \mu^\varphi q_1^\varphi\wt \Phi = q_1^\varphi \wt \Phi $ and use $v^2 \le C(1-\Delta)$. Then,
\begin{align*}
 \big\vert \lsp \Phi,   q_{1}^\varphi q_{2}^\varphi v_{12}   q_{1}^\varphi q_{2}^\varphi   \wt \Phi\rsp \big\vert = &   \big\vert \lsp \widehat m^\varphi  \Phi,  q_{1}^\varphi q_{2}^\varphi v_{12} q_{1}^\varphi q_{2}^\varphi \widehat \mu^\varphi \wt \Phi  \rsp \big\vert  \\
\le  &  \no  \widehat m^\varphi  q_{1}^\varphi q_{2}^\varphi \Phi \no \ \no v_{12} q_{1}^\varphi q_{2}^\varphi \widehat \mu^\varphi \wt \Phi \no \nonumber\\ 
\le &  C  \no   q_{1}^\varphi q_{2}^\varphi q_{3}^\varphi \Phi \no  \ \no \nabla_1 q_{1}^\varphi  \wt \Phi \no  \le \frac{N C \no q_{1}^\varphi q_{2}^\varphi q_{3}^\varphi  \Phi \no^2}{2} + \frac{\no \nabla_1 q_{1}^\varphi \wt \Phi \no^2}{2N}.
\end{align*}
3. Similarly as in 1, we find
\begin{align*}
& \big\vert \lsp  \Phi,  q_{1}^\varphi q_{2}^\varphi A_{12}  q_{1}^\varphi p_{2}^\varphi \wt \Phi \rsp \big\vert \\
& =   \vert \lsp  \widehat f_{\text{odd}}^\varphi  \Phi, q_{1}^\varphi q_{2}^\varphi A_{12}  q_{1}^\varphi p_{2}^\varphi \widehat f_{\text{even}}^\varphi  \wt \Phi \rsp + \lsp  \widehat f_{\text{even}}^\varphi  \Phi,  q_{1}^\varphi q_{2}^\varphi A_{12}  q_{1}^\varphi p_{2}^\varphi \widehat f_{\text{odd}}^\varphi \wt \Phi \rsp \big\vert \\
& =   \big\vert \lsp  \widehat f_{\text{odd}}^\varphi  \Phi,  \widehat \mu^\varphi  q_{1}^\varphi q_{2}^\varphi A_{12} \q{1}{} \p{2}{} \widehat { \tau_{-1} m}^\varphi  \widehat f_{\text{even}}^\varphi \wt \Phi \rsp + \lsp   \widehat f_{\text{even}}^\varphi \Phi, \widehat {\tau_{1} n}^\varphi q_{1}^\varphi q_{2}^\varphi A_{12} q_{1}^\varphi p_{2}^\varphi \widehat \nu^\varphi \widehat {f}_{\text{odd}}^\varphi  \wt \Phi \rsp \big\vert  \\
& \le   \no A_{12} p_{1}^\varphi \no \Big( \no  \widehat {\mu}^\varphi q_{1}^\varphi q_{2}^\varphi \widehat {f}^\varphi_{\text{odd}}  \Phi \no \ \no  q_1^\varphi \widehat { \tau_{-1} m}^\varphi  \wt \Phi \no  + \no  \widehat {\tau_{1} n}^\varphi  q_{1}^\varphi q_{2}^\varphi \Phi \no \ \no  q_{1}^\varphi \widehat \nu^\varphi \widehat {f}^\varphi_{\text{odd}}  \wt \Phi  \no \Big) \\
& \le   \frac{\no \widehat {f}^\varphi_{\text{odd}}  \Phi \no^2 }{2N} + \frac{\no \widehat {f}^\varphi_{\text{odd}}  \wt \Phi \no^2 }{2 N } + C^{\varphi} N \Big(  \no q_{1}^\varphi q_{2}^\varphi q_{3}^\varphi \wt \Phi \no^2 + \no q_{1}^\varphi q_{2}^\varphi q_{3}^\varphi  \Phi \no^2 + \frac{1}{N^3} \Big) \\
& \le   \frac{ \no \widehat {f}^\varphi_{\text{odd}}  \wt \Phi \no^2 }{N} + \frac{ \no \Phi - \wt \Phi \no^2 }{N} + C^{\varphi} N \Big(  \no q_{1}^\varphi q_{2}^\varphi q_{3}^\varphi \wt \Phi \no^2 + \no q_{1}^\varphi q_{2}^\varphi q_{3}^\varphi \Phi \no^2 + \frac{1}{N^3} \Big) .
\end{align*} 
The term containing the hermitian conjugate is estimated in exactly the same manner.
\end{proof}

\begin{proof}[Proof of Lemma \ref{TIME:DERIVATIVE:ALPHA}] We prove only the first part of the lemma since the second part is proved analogously. Using \eqref{TIME:DERIVATIVE:CF} and the symmetry of the wave function,
\begin{align*}
\partial_t  \lsp \Psi_t, \cf{f}{t} \Psi_t \rsp    
= &  i N \lsp \Psi_t, \big( Z^{\varphi_t}(x_1-x_2) \cf{f}{t} - \cf{f}{t} Z^{\varphi_t}(x_1-x_2)\big) \Psi_t \rsp . \nonumber
\end{align*}
Multiplying both of the $\Psi_t$ with the identity $1 = (\p{1}{t}+\q{1}{t})(\p{2}{t}+\q{2}{t})$, leads to
\begin{align}
 \partial_t  \lsp \Psi_t, \cf{f}{t} \Psi_t \rsp  =  
&   2 i N \lsp \Psi_t, \p{1}{t}  \p{2}{t} \big( Z^{\varphi_t}(x_1-x_2) \cf{f}{t} - \cf{f}{t} Z^{\varphi_t}(x_1-x_2)\Big) \q{1}{t}  \p{2}{t} \Psi_t \rsp + \text{c.c}\nonumber\\
& +  i N \lsp \Psi_t, \p{1}{t}  \p{2}{t} \big( Z^{\varphi_t}(x_1-x_2) \cf{f}{t} - \cf{f}{t} Z^{\varphi_t}(x_1-x_2)\Big)  \q{1}{t}  \q{2}{t}\Psi_t \rsp + \text{c.c} \nonumber\\
& + 2 i N \lsp \Psi_t, \q{1}{t}  \p{2}{t} \Big( Z^{\varphi_t}(x_1-x_2) \cf{f}{t} - \cf{f}{t} Z^{\varphi_t}(x_1-x_2)\big)  \q{1}{t}  \q{2}{t} \Psi_t \rsp + \text{c.c}, \nonumber
\end{align}
where $\text{c.c.}$ denotes the complex conjugate of the preceding expression (note that due to the symmetry of $\Psi_t$, all other contributions are zero).\ Applying the pull through formula \eqref{PULL:THROUGH:FORMULA} proves the first part of the lemma. 
\end{proof}
\begin{proof}[Proof of Lemma \ref{UNIFORM:ROEMISCHE:1-3}]
\textit{Term} (I): The first term is identically zero for all $n\in \mathbb N$,
\begin{align*}
  \text{(I)}_{m^n,\Psi}   = & 2N \im \vert \lsp \Psi, \p{1}{t}  \big(  \p{2}{t} v(x_1-x_2)  \p{2}{t} -  \p{2}{t} \big( v\ast \vert \varphi_t  \vert^2\big)(x_1)  \p{2}{t}\big) \q{1}{t} \Big( ( \cf{m}{t})^n - (\cfs{\tau_{-1} m}{t}{} )^n \Big) \Psi \rsp = 0,
\end{align*}
because 
\begin{align}
 \p{2}{t}  v(x_1-x_2) \p{2}{t}  =  \p{2}{t} \big(v\ast \vert \varphi_t  \vert^2 \big)(x_1) \p{2}{t} \label{pvp:MEAN:FIELD}
\end{align}
cancels exactly the mean field potential (it is this term which determines the choice of the effective potential in the Hartree equation).\\
\\
For the second and third term, we compute (using the binomial expansion for $(k-d)^l$)
\begin{align*}
(\cf{m}{t})^n - (\cf{\tau_{-d}m}{t}{})^n =&  \sum_{k=0}^N \Big[ \Big(\frac{k}{N}\Big)^n - \Big(\frac{k-d}{N}\Big)^n \Big] P_{N,k}^{\varphi_t}  = \sum_{l=0}^{n-1} \frac{C^{d,l}}{N^{n-l}} (\cf{m}{t})^l
\end{align*}
for some constants $C^{d,l}$.\\
\\
\textit{Term} (II):
It follows that (note that $\p{2}{t}\bar v_1^{\varphi_t} \q{2}{t} = 0$) 
\begin{align}
\vert \text{(II)}_{m^n,\Psi} \vert = & \big\vert \sum_{l=0}^{n-1} \frac{C^{d,l}}{N^{n-l-1}} \lsp   \Psi, \p{1}{t}  \p{2}{t}  v(x_1-x_2)  \q{1}{t}  \q{2}{t}  (\cf{m}{t})^l  \Psi \rsp \big\vert   \nonumber\\
= & \big\vert \sum_{l=0}^{n-1} \frac{C^{d,l}}{N^{n-l-1}} \lsp (\cfs{\tau_{-2}m}{t}{})^{\frac{l}{2}} \cfs{\tau_{-2}n}{t}{} \Psi, \p{1}{t}  \p{2}{t}  v(x_1-x_2)  \q{1}{t}  \q{2}{t} (\cf{m}{t})^{ \frac{l}{2}} \cf{\nu}{t}  \Psi \rsp \big\vert   \nonumber\\
\le &   \sum_{l=0}^{n-1} \frac{\vert C^{d,l} \vert }{N^{n-l-1}} \no v_{12} \p{1}{t}  \p{2}{t} \no_{op}\ \no (\cfs{\tau_{-2}m}{t}{})^{\frac{l+1}{2}} \Psi\no \ \no  \q{1}{t}  \q{2}{t} (\cf{m}{t})^{ \frac{l}{2}} \cf{\nu}{t} \Psi\no \nonumber\\
\le &   \sum_{l=0}^{n-1}  \frac{C^{\varphi_t}}{N^{n-l-1}} \sqrt{\sum_{j=0}^{l+1} {l+1 \choose j} \Big(\frac{2}{N}\Big)^{l+1-j} \lsp \Psi,   (\cfs{m}{t}{})^j \Psi \rsp }  \sqrt{\lsp \Psi,   \q{1}{t}  \q{2}{t} (\cf{m}{t})^{l} \cf{\mu}{t} \Psi\rsp } \nonumber\\
\le &  \sum_{l=0}^{n-1}  \frac{C^{\varphi_t}}{N^{n-l-1}}  \sum_{j=0}^{l+1} {l+1 \choose j} \Big(\frac{2}{N}\Big)^{l+1-j} \lsp \Psi,   (\cfs{m}{t}{})^j \Psi\rsp  +  \sum_{l=0}^{n-1}  \frac{C^{\varphi_t}}{N^{n-l-1}}  \lsp \Psi_t,     (\cf{m}{t})^{l+1}\Psi_t\rsp \nonumber\\
\le & \sum_{l=0}^{n} \frac{C^{\varphi_t} \lsp \Psi, (\cf{m}{t})^{l}\Psi\rsp}{N^{n-l}}.
\end{align}
The essential ingredient here is the symmetry of the wave function which ensures that not all mass can be located around, e.g., $x_1\approx x_2$ (for general $\Psi\in L^2(\mathbb R^{3N})$, the second term would not be necessarily small).
\\
\\
\textit{Term} (III):
Again via the pull through formula, and similarly as in $\text{(II)}$,
\begin{align*}
\vert \text{(III)}_{m^n,\Psi} \vert  = & \big\vert  \sum_{l=0}^{n-1} \frac{C^{d,l}}{N^{n-l-1}} \lsp \Psi, \q{1}{t}  \p{2}{t} \big( v(x_1-x_2) - v\ast \vert \varphi_t\vert^2 (x_1) \big) \q{1}{t} \q{2}{t}  (\cf{m}{t})^l  \Psi \rsp \big\vert \\
 = &\big\vert  \sum_{l=0}^{n-1} \frac{C^{d,l}}{N^{n-l-1}} \lsp (\cf{\tau_{-1}m}{t})^{\frac{l}{2} } \Psi, \q{1}{t}  \p{2}{t} \big( v(x_1-x_2) - v\ast \vert \varphi_t\vert^2(x_1) \big) \q{1}{t} \q{2}{t}  (\cf{m}{t})^{\frac{l}{2} }  \Psi \rsp \big\vert   \\
 \le    & \sum_{l=0}^{n-1} \frac{C^{\varphi_t}}{N^{n-l-1}} \no \q{1}{t} (\cf{\tau_{-1}m}{t})^{\frac{l}{2} } \Psi \no \ \no \q{1}{t}   (\cf{m}{t})^{\frac{l}{2}}  \Psi \no \\
\le & \sum_{l=0}^{n} \frac{C^{\varphi_t} \lsp \Psi, (\cf{m}{t})^{l}\Psi \rsp}{N^{n-l}}.
\end{align*}
\end{proof}

\begin{proof}[Proof of Lemma \ref{A:PRIORI:ESTIMATE:PSI}]\ By Lemma \ref{PROPOSITION:Q:K:M:N} for $k=0$, it is sufficient to estimate the value of $$\lsp  \Psi_t , (\cf{m}{t})^n \Psi_t \rsp, \hspace{1cm} m(k)=\frac{k}{N}.$$ To this end, we compute its time derivative and conclude via a Gr\"onwall argument.\ Recall that by Lemma \ref{TIME:DERIVATIVE:ALPHA}, we have
\begin{align}
\partial_t \lsp  \Psi_t , (\cf{m}{t})^n \Psi_t \rsp = & \text{(I)}_{m^n,\Psi_t} +\text{(II)}_{m^n,\Psi_t} +\text{(III)}_{m^n,\Psi_t}  \label{TIME:DERIVATIVE:M:TO:N},\\
\partial_t \lsp \wt \Psi_t , (\cf{m}{t})^n \wt \Psi_t \rsp = & \text{(II)}_{m^n,\wt \Psi_t}  \label{TIME:DERIVATIVE:M:TO:N:2} .
\end{align}
The terms on the r.h.s.\ have been estimated in Lemma \ref{UNIFORM:ROEMISCHE:1-3}. The remainder of the argument follows by induction. Assume that for all $k\le n-1$,
\begin{align}
\lsp \Psi_t, (\cf{m}{t})^k \Psi_t \rsp \le 
e^{C^{\varphi_t}} \sum_{l=0}^{k} \frac{C_{n,k}}{N^{l-n}}   \lsp \Psi_0, (\cf{m}{0})^{l}\Psi_0 \rsp \label{INDUKTIONS:ANFANG}.
\end{align}
By means of \eqref{TIME:DERIVATIVE:M:TO:N} and Lemma \ref{UNIFORM:ROEMISCHE:1-3},
\begin{align*}
\partial_t  \lsp \Psi_t, (\cf{m}{t})^n \Psi_t \rsp \le   C^{\varphi_t} \lsp \Psi_t, (\cf{m}{t})^n \Psi_t \rsp + e^{C^{\varphi_t}} \sum_{l=0}^{n-1} \frac{C_{n,l}}{N^{l-n}}   \lsp \Psi_0, (\cf{m}{0})^{l}\Psi_0\rsp
\end{align*}
which implies by Gr\"onwall's inequality that
\begin{align*}
\lsp \Psi_t, (\cf{m}{t})^n \Psi_t \rsp \le e^{C^{\varphi_t}} \sum_{l=0}^{n} \frac{C_{n,l}}{N^{l-n}}   \lsp \Psi_0, (\cf{m}{0})^{l}\Psi_0\rsp.
\end{align*}
The case $n=2$ follows as well from Gr\"onwall, cf.\ again Lemma \ref{UNIFORM:ROEMISCHE:1-3},
\begin{align*}
\partial_t  \lsp \Psi_t, \cf{m}{t} \Psi_t \rsp \le   C^{\varphi_t} \Big( \lsp \Psi_t, \cf{m}{t} \Psi_t \rsp + \frac{1}{N} \Big) \ \ \ \Rightarrow  \ \ \ \lsp \Psi_t, \cf{m}{t} \Psi_t \rsp \le   e^{C^{\varphi_t}} \Big( \lsp \Psi_0, \cf{m}{0} \Psi_0 \rsp + \frac{1}{N} \Big).
\end{align*}
Starting from \eqref{TIME:DERIVATIVE:M:TO:N:2}, the argument is exactly the same for $\wt \Psi_t$.
\end{proof}

\begin{proof}[Proof of Lemma \ref{LEMMA:ODD:EVEN:PART:TIME:EVOLUTION}]
The time-derivative of $\no \cfs{f}{t}{\text{odd}} \wt \Psi_t \no^2$ is given by
\begin{align*}
\partial_t \lsp \wt \Psi_t, \cfs{f}{t}{\text{odd}} \wt \Psi_t \rsp =  \text{(II)}_{f_{\text{odd}},\wt \Psi_t} = 2 N  \im \lsp \wt  \Psi_t, \p{1}{t}  \p{2}{t}  v_{12}  \q{1}{t}  \q{2}{t} \big( \cfs{f}{t}{\text{odd}} - \cfs{\tau_{-2} f}{t}{\text{odd}} \big)  \wt \Psi_t \rsp,
\end{align*}
cf.\ Lemma \ref{TIME:DERIVATIVE:ALPHA}. Recalling the definition of the shifted weight function,
\begin{align*}
\cfs{f}{t}{\text{odd}} - \cfs{\tau_{-2} f}{t}{\text{odd}} = f(1) P_{N,1}^{\varphi_t} =  P_{N,1}^{\varphi_t},
\end{align*}
and the fact that 
$$ \q{1}{t}  \q{2}{t}  P_{N,1}^{\varphi_t} = 0$$
shows that $\no \cfs{f}{t}{\text{odd}} \wt \Psi_t \no = \no \cfs{f}{0}{\text{odd}} \wt \Psi_0 \no$. A similar calculation holds for the even case.
\end{proof}

\begin{proof}[Proof of Lemma \ref{LEMMA:ENERGY:ESTIMATE:PSI:TILDE}] The proof consists of two steps. First, we show that the bad part of the kinetic energy can be bounded as follows:
\begin{align}
 \no \nabla_1 q_1 \Psi_t \no^2 \le C \no \Psi_t -\widetilde \Psi_t \no^2 + C \vert \mathcal E_{\widetilde H^t}(\Psi_t) - \mathcal E_{h^{t,\varphi_t}}(\varphi_t) \vert + C^{\varphi_t} \Big(\no \cfs{f}{0}{\text{odd}} \Psi_0 \no^2 +  \no \q{1}{t}  \Psi_t \no^2  +\frac{1}{N} \Big) \label{BAD:KIN:E:INEQUALITY}
\end{align}
where 
$$ \mathcal E_{\widetilde H^t}(\Psi) = \no \nabla_1 \Psi \no^2 + \lsp \Psi,  ( W^t_1 + \bar v_1^{\varphi_t} ) \Psi \rsp + \lsp \Psi, \wt v^t_{12} \Psi \rsp $$ denotes the energy per particle w.r.t.\ to $\wt H^t$. Here, and below, we are using the abbreviation $\wt v_{12}^t = \big( \p{1}{t}\q{2}{t} v_{12} \q{1}{t}\p{2}{t} + \p{1}{t}\p{2}{t} v_{12} \q{1}{t}\q{2}{t} \big) + \text{h.c.}$ In the second step, we use energy conservation of $H^t$ (modulo the change due to the external potential $W^t$) in order to show that the energy difference that appears on the r.h.s.\ in the above line can be approximated in terms of
\begin{align}
\vert \mathcal E_{\widetilde H^t}(\Psi_t) - \mathcal E_{h^{t,\varphi_t}}(\varphi_t) \vert  \le & \vert {\mathcal E}_{\widetilde H^t}(\Psi_t)  -  {\mathcal E}_{H^t} (\Psi_t) \vert  + \vert {\mathcal E}_{H^0}(\Psi_0)  -  {\mathcal E}_{h^{0,\varphi_0}}(\varphi_0) \vert \nonumber \\
& ~~~~~ + C^{\varphi_t} \Big(\no   \cfs{f}{0}{\text{odd}} \Psi_0 \no^2 + N  \no \q{1}{t} \q{2}{t} \Psi_t \no^2 +   \no \q{1}{t}  \Psi_t \no^2  +\frac{1}{N} \Big).\label{ENERGY:DIFFERENCE:INITIAL:TIME}
\end{align}
To obtain the first inequality, one inserts the identity $1 = (\p{1}{t} + \q{1}{t} )(\p{2}{t} + \q{2}{t} )$ on the left and right hand side of $\widetilde H_t$ in ${\mathcal E}_{\widetilde H^t}(\Psi_t)$ and extracts the bad part of the kinetic energy (the first term on the r.h.s.):
\begin{align}
 {\mathcal E}_{\widetilde H^t}(\Psi_t) - \mathcal E_{h^{t,\varphi_t}}(\varphi_t) = & \lsp i\nabla_1 \q{1}{t} \Psi_t,i\nabla_1\q{1}{t}  \Psi_t \rsp \nonumber \\
 & + \lsp  \q{1}{t} \Psi_t,  ( W^t_1 + \bar v_1^{\varphi_t} ) \q{1}{t}  \Psi_t \rsp  \label{ENERGY:ESTIMATE:PSI:TILDE:01}\\ 
& + \lsp \Psi_t, \p{1}{t}  h_1^{t,\varphi_t} \p{1}{t}  \Psi_t \rsp - \langle \varphi_t, h_1^{\varphi_t} \varphi_t \rangle \label{ENERGY:ESTIMATE:PSI:TILDE:02}\\
& + 2\re \lsp   \Psi_t, \p{1}{t}  h_1^{t,\varphi_t} \q{1}{t} \Psi_t \rsp \label{ENERGY:ESTIMATE:PSI:TILDE:03} \\
& +  \lsp  \Psi_t, \p{1}{t}\p{2}{t}  \widetilde v_{12}^t \p{1}{t} \p{2}{t}  \Psi_t \rsp \label{ENERGY:ESTIMATE:PSI:TILDE:07}\\
& +  \lsp  \Psi_t, (1-\p{1}{t}\p{2}{t}) \widetilde v_{12}^t (1- \p{1}{t} \p{2}{t})   \Psi_t \rsp \label{ENERGY:ESTIMATE:PSI:TILDE:08}\\
& + 2\re \lsp  \Psi_t, (1-\p{1}{t}\p{2}{t}) \widetilde v_{12}^t \p{1}{t}\p{2}{t} \Psi_t \rsp \label{ENERGY:ESTIMATE:PSI:TILDE:09},
\end{align}
All but the first line on the r.h.s.\ can be estimated separately:
\begin{align*}
\vert \eqref{ENERGY:ESTIMATE:PSI:TILDE:01} \vert \le  & (  \no W^t_1 \no_{\infty} + \no \bar v_1^{\varphi_t} \no_{\infty} ) \no \q{1}{t}   \Psi_t \no^2,\\
\vert \eqref{ENERGY:ESTIMATE:PSI:TILDE:02} \vert  = & \vert \lsp  \Psi_t,\big( \p{1}{t} h^{t,\varphi_t}_1 \p{1}{t}  - \langle \varphi_t, h^{t,\varphi_t}_1 \varphi_t \rangle \big)  \Psi_t \rsp \vert \le  \no h^{\varphi_t}_1 \varphi_t \no \ \no \q{1}{t}   \Psi_t \no^2, \\
\vert \eqref{ENERGY:ESTIMATE:PSI:TILDE:07} \vert =& 0, \\
\vert \eqref{ENERGY:ESTIMATE:PSI:TILDE:08}\vert =  & \vert \lsp  \Psi_t, (1-\p{1}{t}\p{2}{t})  \wt v^t_{12}(1- \p{1}{t} \p{2}{t})  \Psi_t \rsp \vert = \vert 2  \lsp  \Psi_t, \q{1}{t}\p{2}{t} v_{12} \p{1}{t} \q{2}{t}  \Psi_t \rsp \vert \\
\le & 2 \no v_{12} \p{2}{t} \no_{op} \no\q{1}{t} \Psi_t \no^2\le 2 \sqrt{\no v^2\ast \vert \varphi_t\vert^2\no_{\infty}}  \no\q{1}{t} \Psi_t \no^2,\\
\vert \eqref{ENERGY:ESTIMATE:PSI:TILDE:09} \vert = & \vert2\re \lsp  \Psi_t, \q{1}{t}\q{2}{t} v_{12} \p{1}{t}\p{2}{t}  \Psi_t \rsp \vert = \vert2\re \lsp  \Psi_t, \cf{\nu}{t}\q{1}{t}\q{2}{t} v_{12} \p{1}{t} \p{2}{t} \cfs{\tau_{-2}n}{t}{}  \Psi_t \rsp \vert\\
\le &   2 \no v_{12} \p{2}{t} \no_{op} \ \no \cfs{\tau_{-2}n}{t}{}  \Psi_t\no \ \no   \cf{\nu}{t}\q{1}{t}\q{2}{t}  \Psi_t\no \le  \sqrt{\no v^2\ast \vert \varphi_t\vert^2\no_{\infty}}  \Big( \no \q{1}{t} \Psi_t \no^2 + \frac{1}{N} \Big).
\end{align*}
In order to estimate the remaining line, we use \eqref{QAP:COMMUTATION} with $A_1 =   h_1^{t,\varphi_t}$,
\begin{align*}
\vert \eqref{ENERGY:ESTIMATE:PSI:TILDE:03}  \vert = & \vert 2\re \lsp  \Psi_t,  \p{1}{t} h_1^{t,\varphi_t} \q{1}{t}  \Psi_t \rsp\vert \le  C \no \Psi_t - \wt \Psi_t \no^2 +  C \no \cfs{f}{t}{\text{odd}} \wt \Psi_t \no^2 + C^{\varphi_t} \no \q{1}{t} \Psi_t \no^2.
\end{align*}
This completes the proof of inequality \eqref{BAD:KIN:E:INEQUALITY}. It the second step we need to estimate the energy difference on the r.h.s.\ of \eqref{BAD:KIN:E:INEQUALITY}. For that, we use the fact that only the time-dependent external potential causes a change in the energy,
\begin{align*}
\partial_t \big(  \mathcal E_{H^t}(\Psi_t) - \mathcal E_{h^{t,\varphi_t}}(\varphi_t) \big) = &  \lsp \Psi_t, \dot  W_1^t \Psi_t \rsp-  \langle \varphi_t, \dot W_1^t \varphi_t \rangle  \\
= & \lsp \Psi_t, \p{1}{t}\dot  W_1^t \p{1}{t} \Psi_t \rsp-  \langle \varphi_t, \dot W_1^t \varphi_t \rangle  \\
& + \lsp \Psi_t, \q{1}{t}\dot  W_1^t \q{1}{t} \Psi_t \rsp + 2 \re \lsp \Psi_t, \p{1}{t}\dot  W_1^t \q{1}{t} \Psi_t \rsp\\
& \le  C^{\varphi_t} \no \q{1}{t} \Psi_t \no^2 + C \no \Psi_t - \wt \Psi_t \no^2 +  C \no \cfs{f}{t}{\text{odd}} \wt \Psi_t \no^2,
\end{align*}
where we have used that $\dot W^t \in L^\infty$ and applied inequality \eqref{QAP:COMMUTATION}. Hence,
\begin{align*}
\vert  \mathcal E_{\widetilde H^t}(\Psi_t) - \mathcal E_{h^{\varphi_t}}(\varphi_t) \vert \le &  \vert  \mathcal E_{\widetilde H^t}(\Psi_t) - \mathcal E_{H^t}(\Psi_t) \vert + \vert  \mathcal E_{H^0}(\Psi_0) - \mathcal E_{h^{\varphi_0}}(\varphi_0) \vert \\
& \hspace{1.5cm} + C^{\varphi_t} \no \q{1}{t} \Psi_t \no^2 + C \no \Psi_t - \wt \Psi_t \no^2 +  C \no \cfs{f}{t}{\text{odd}} \wt \Psi_t \no^2  .
\end{align*}
By \eqref{DIFFERENCE:HAMILTONIAN}, the first term on the r.h.s.\ is given by
\begin{align*}
\vert  \mathcal E_{\widetilde H^t}(\Psi_t) - \mathcal E_{H^t}(\Psi_t) \vert & = \Big\vert 2 \lsp \Psi_t, \Big[ \Big( \q{1}{t} \q{2}{t} (v_{12}-\bar v_1^{\varphi_t} )  \q{1}{t} \p{2}{t} +  \q{1}{t} \q{2}{t} (v_{12}-\bar v_1^{\varphi_t} + \mu^{\varphi_t})  \q{1}{t} \q{2}{t} \Big) + \text{h.c.} \Big] \Psi_t \rsp \Big\vert \\
& \le C \big( \no \bar v_1^{\varphi_t} \no_{\infty} + \mu^{\varphi_t} \big) \no \q{1}{t} \Psi_t\no^2 + 2 \big\vert \lsp \Psi_t, \q{1}{t} \q{2}{t} v_{12} \q{1}{t} \q{2}{t} \Psi_t\rsp \big\vert ,
\end{align*}
and, by means of $v^2 \le C(1-\Delta)$,
\begin{align*}
\big\vert \lsp \Psi_t, \q{1}{t} \q{2}{t} v_{12} \q{1}{t} \q{2}{t} \Psi_t\rsp \big\vert \le C N \no \q{1}{t} \q{2}{t} \Psi_t \no^2 + \frac{C}{N} \no v_{12} \q{1}{t} \q{2}{t} \Psi_t \no^2 \le C N \no \q{1}{t} \q{2}{t} \Psi_t \no^2  + \frac{C}{N} \no \nabla_1 \q{1}{t} \Psi_t \no^2.
\end{align*}
This completes the proof of the lemma.
\end{proof}

\begin{proof}[Proof of Lemma \ref{LEMMA:CHI:TILDE:CHI:BOG}]  
Using the decomposition in \eqref{DECOMPOSITION:PSI:TILDE}, it can be verified by direct calculation that if $\widetilde \Psi_t$ solves the equation $ i\partial_t \widetilde \Psi_t = \widetilde H_t \widetilde \Psi_t$,
then the corresponding $(\widetilde \chi^{(k)}_t)_{k=0}^N$ are solutions to the following system of coupled equations.
\begin{align*}
i\partial_t  \wt \chi^{(0)}_t = &  \sqrt{\frac{N}{N-1}} A^{(2\to 0),t} \wt \chi^{(2)}_t  \\
i\partial_t  \wt \chi^{(1)}_t = &  \Big( h^{t,\varphi_t}  + K^{(1),t} \Big) \wt \chi^{(1)}_t + \sqrt{\frac{N-2}{N-1}} A^{(3\to 1),t}  \wt \chi^{(3)}_t,\\
  i\partial_t  \wt \chi^{(k)}_t = &   \sum_{i=1}^k  \Big( h^{t,\varphi_t}_i + \frac{N-k}{N-1} K^{(1),t}_{i} \Big) \wt \chi_t^{(k)}  \\
  &  +  \frac{\sqrt{(N-k+2)(N-k+1)}}{N-1} A^{(k-2\to k),t} \wt \chi^{(k-2)}_t  +   \frac{\sqrt{(N-k)(N-k-1)}}{N-1}A^{(k +2 \to k),t*} \wt \chi_t^{(k+2)}
\end{align*}
for all $2\le k \le N$. Here, we have introduced the abbreviations
\begin{align*}
 		A^{(k-2\to k),t} \wt \chi^{(k-2)} = & \frac{1}{2\sqrt{k(k-1)}} \sum_{1\le i<j\le k} K^{(2),t}(x_i,x_j)  \chi^{(k-2)}(x_1,...,x_k\backslash x_i \backslash x_j ) \in \mathcal H_s^{t,k} , \\  	
 		A^{(k+2\to k),t*} \chi^{(k+2)} = & \frac{\sqrt{(k+1)(k+2)}}{2} \int  \int  \overline{ K^{(2),t}(x,y) } \chi^{(k+2)}(x_1,...,x_k,x,y)  dxdy  \in \mathcal H_s^{t,k} ,
\end{align*}
which for wave functions $\phi^{(k)} \in  \mathcal H_s^{(k),t}$, $\chi^{(k-2)} \in \mathcal H_s^{(k-2),t} $ satisfy the relation
\begin{align}
\lsp \phi^{(k)}, A^{(k-2\to k),t}  \chi^{(k-2)} \rsp_{\mathcal H_s^{(k),t}} =  \lsp A^{(k\to k-2),t*}  \phi^{(k)} , \chi^{( k-2 )} \rsp_{\mathcal H_s^{(k-2),t}}. \label{ADJOINT:A:K}
\end{align}
We further set $(\widetilde \chi^{(k)}_t=0)_{k\ge N+1}$ for all $t\ge 0$.
The next step is to bound the series 
\begin{align*}
\sum_{k=0}^\infty   \no \widetilde \chi_t^{(k)} - \chi_t^{(k)} \no^2_{\mathcal H_s^{(k),t}} \le \sum_{k=0}^\infty    \no \widetilde \chi_t^{(k)} \no^2_{\mathcal H_s^{(k),t} } + \sum_{k=0}^\infty \no \chi_t^{(k)} \no^2_{\mathcal H_s^{(k),t}} <\infty
\end{align*}
which is finite since $\sum_{k=0}^\infty    \no \widetilde \chi_t^{(k)} \no^2_{\mathcal H^{(k),t}_s} = \no \wt \Psi_t \no^2 = 1$ and since the Bogoliubov hierarchy is well-posed in the sense that $\sum_{k=0}^\infty \no \chi_t^{(k)} \no^2_{\mathcal H^{(k),t}_s} <\infty$, see in \cite[Section 4.3]{Lewin:2015a}. One thus finds (note that due to finiteness of the series one can differentiate it termwise),
\begin{align}
&  \frac{1}{2}\sum_{k=0}^\infty \partial_t  \no \widetilde \chi_t^{(k)} - \chi_t^{(k)} \no^2_{\mathcal H^{(k),t}_s} \nonumber\\
= &-  \sum_{k=1}^N \im \lsp \widetilde \chi_t^{(k)} - \chi_t^{(k)},  \sum_{i=1}^k \Big( h^{t,\varphi_t}_i  + K^{(1),t}_{i} \Big) \big( \widetilde \chi_t^{(k)} - \chi_t^{(k)} \big) \rsp_{\mathcal H^{(k),t}_s } \label{DERIVATIVE:DIFFERENCE:CHI:BOG:1}  \\
& - \sum_{k=1}^{N} \im \lsp \widetilde \chi_t^{(k)} - \chi_t^{(k) }, \sum_{i=1}^k \Big( \frac{N-k}{N-1} -1 \Big)
K^{(1),t}_{i} \widetilde \chi_t^{(k)} \rsp_{\mathcal H^{(k),t}_s}  \label{DERIVATIVE:DIFFERENCE:CHI:BOG:2}\\ 
&-  \sum_{k=N+1}^\infty \im \lsp  \chi_t^{(k)}, \sum_{i=1}^k \Big( h^{t,\varphi_t}_i  + K^{(1),t}_{i}   \Big) \chi_t^{( k ) }  \rsp_{ \mathcal H^{(k),t}_s } \label{DERIVATIVE:DIFFERENCE:CHI:BOG:3} \\
& -  \sum_{k=2}^{N} \im \lsp \widetilde \chi_t^{(k)} - \chi_t^{( k ) }, {A^{(k-2\to k),t* }} \big( \widetilde \chi_t^{(k-2) } -  \chi_t^{( k-2) } \big) \rsp_{\mathcal H^{(k),t}_s} \label{DERIVATIVE:DIFFERENCE:CHI:BOG:4} \\
& -  \sum_{k=2}^{N} \im \lsp \widetilde \chi_t^{(k)} - \chi_t^{( k ) }, \Big( \frac{\sqrt{(N-k+2)(N-k+1)}}{N-1} -1 \Big)  A^{(k-2\to k),t *}	\widetilde \chi_t^{( k-2 ) }   \rsp_{\mathcal H^{(k),t}_s}  \label{DERIVATIVE:DIFFERENCE:CHI:BOG:5}\\
& -  \sum_{k=0}^{N-2} \im \lsp \widetilde \chi_t^{(k)} - \chi_t^{( k ) },  A^{(k+2\to k),t}  \big( \widetilde \chi^{ ( k+2 ) }  -  \chi_t^{ ( k+2 ) } \big) \rsp_{\mathcal H^{(k),t}_s }  \label{DERIVATIVE:DIFFERENCE:CHI:BOG:6}\\
& -  \sum_{k=0}^{N-2} \im \lsp \widetilde \chi_t^{(k)} - \chi_t^{( k )}, 
\Big(  \frac{\sqrt{(N-k)(N-k-1)}}{N-1} - 1 \Big) A^{(k+2\to k),t} \widetilde \chi_t^{( k+2 ) }  \rsp_{ \mathcal H^{(k),t}_s } \label{DERIVATIVE:DIFFERENCE:CHI:BOG:7} \\
& + \sum_{k=N-1}^{N} \im \lsp \widetilde \chi_t^{(k)} - \chi_t^{( k ) }, 
 A^{(k+2\to k ),t} \chi_t^{( k+2 ) }  \rsp_{\mathcal H^{(k),t}_s} \label{DERIVATIVE:DIFFERENCE:CHI:BOG:8} \\
& -  \sum_{k=N+1}^{\infty} \im \lsp  \chi_t^{( k )}, A^{(k-2\to k),t* } \chi_t^{( k-2 )}   \rsp_{ \mathcal H^{(k),t}_s }  \label{DERIVATIVE:DIFFERENCE:CHI:BOG:9}\\
& -  \sum_{k=N-1}^{\infty} \im \lsp \chi_t^{( k )}, A^{(k+2 \to k),t} \chi_t^{( k+2 )}   \rsp_{\ \mathcal H^{(k),t}_s}  \label{DERIVATIVE:DIFFERENCE:CHI:BOG:10}.
\end{align}
We first note that
$$ \eqref{DERIVATIVE:DIFFERENCE:CHI:BOG:1} = 0, \hspace{1cm} \eqref{DERIVATIVE:DIFFERENCE:CHI:BOG:3}=0$$
since ${K^{(1),t*}} = K^{(k),t}$. Using Cauchy Schwarz as well as the geometric mean inequality, one finds
\begin{align*}
\vert \eqref{DERIVATIVE:DIFFERENCE:CHI:BOG:2} \vert \le & \sum_{k=1}^N \no \widetilde \chi_t^{(k)} - \chi_t^{( k ) } \no^2_{\mathcal H_s^{(k),t} } + \sum_{k=1}^N \Big( 1- \frac{N-k}{N-1} \Big)^2 \lno \sum_{i=1}^k K^{(1),t}_i \rno_{op}^2 \  \no \widetilde \chi_t^k\no^2_{\mathcal H_s^{(k),t}} \\
\le & \sum_{k=0}^\infty \no \widetilde \chi_t^k - \chi_t^{(k)} \no^2_{\mathcal H_s^{(k),t}} + C^{\varphi_t}  \sum_{k=0}^N \frac{k^4}{N^2}  \no \widetilde \chi_t^{(k)}\no^2_{\mathcal H_s^{(k),t}},
\end{align*}
where we have used $\no K^{(1),t}_i \no_{op}^2 \le C^{\varphi_t}$. Recalling that $\no \widetilde \chi_t^{(k)}\no^2_{\mathcal H_s^{(k),t}} = \no P_{N,k}^{\varphi_t}\widetilde \Psi_t \no^2$, we can estimate the last factor in the second summand by means of Lemma \ref{A:PRIORI:ESTIMATE:PSI}, i.e.,
\begin{align}\label{ALPHA:ESTIMATE:FLUCTUATIONS}
\sum_{k=0}^N \frac{k^4}{N^2}  \no \widetilde \chi_t^{(k)}\no^2_{\mathcal H_s^{(k),t} } \le   N^2 \sum_{k=0}^N \frac{k^3}{N^3} \no P_{N,k}^{\varphi_t}\widetilde \Psi_t \no^2 = N^2 \lsp \wt \Psi_t, (\cf{m}{t})^3 \wt \Psi_t\rsp \le \frac{e^{C^{\varphi_t}}}{N}.
\end{align}
Recalling \eqref{ADJOINT:A:K}, and substituting the summation index,
\begin{align*}
\eqref{DERIVATIVE:DIFFERENCE:CHI:BOG:4} + \eqref{DERIVATIVE:DIFFERENCE:CHI:BOG:6} = & -  \sum_{k=0}^{N-2} \im \lsp A^{(k+2\to k),t} \big( \widetilde \chi_t^{(k+2)} - \chi_t^{ (k +2 ) }\big),   \widetilde \chi_t^{( k ) } -  \chi_t^{ ( k ) }  \rsp_{\mathcal H_s^{(k),t} }  \\
& -  \sum_{k=0}^{N-2} \im \lsp \widetilde \chi_t^{(k)} - \chi_t^{ ( k ) }, A^{(k+2\to k),t} \big( \widetilde \chi_t^{( k+2 ) }  -  \chi_t^{ ( k+2 ) } \big) \rsp_{\mathcal H_s^{(k),t} }  = 0.
\end{align*}
The same way, one finds
\begin{align*}
\eqref{DERIVATIVE:DIFFERENCE:CHI:BOG:9} + \eqref{DERIVATIVE:DIFFERENCE:CHI:BOG:10} = 0.
\end{align*}
Next, using again Cauchy Schwarz and the geometric mean inequality, and in the second step $\no {A^{(k-2\to k),t*}} \no^2_{op} \le C^{\varphi_t} k^2$ and $(\frac{\sqrt{(N-k+2)(N-k+1)}}{N-1} -1)^2 \le C \frac{k^2}{N^2}$,
\begin{align*}
\vert \eqref{DERIVATIVE:DIFFERENCE:CHI:BOG:5} \vert \le  &  \sum_{k=2}^{N} \no  \widetilde \chi_t^{(k)} - \chi_t^{( k ) } \no_{\mathcal H_s^{(k),t} }^2 + \sum_{k=2}^{N}   \  \Big( \frac{\sqrt{(N-k+2)(N-k+1)}}{N-1} -1 \Big)^2  \no {A^{(k-2\to k),t*}} \no^2_{op} \ \no \widetilde \chi_t^{( k-2 ) } \no_{\mathcal H^{(k),t}_s}^2 \\
\le &  \sum_{k=0}^{\infty} \no  \widetilde \chi_t^{(k)} - \chi_t^{(k)} \no_{\mathcal H_s^{(k),t} }^2 + C^{\varphi_t} \sum_{k=2}^{N}   \  \frac{k^4}{N^2}  \no P_{N,k-2}^{\varphi_t} \wt \Psi_t \no^2 \\
\le &  \sum_{k=0}^{\infty} \no  \widetilde \chi_t^{(k)} - \chi_t^{( k ) } \no_{\mathcal H_s^{(k),t} }^2 + C^{\varphi_t} \sum_{k=0}^{N-2}   \  \frac{(k^4+k^3+k^2+k+1)}{N^2}  \no P_{N,k}^{\varphi_t} \wt \Psi_t \no^2   \\
\le & \sum_{k=0}^{\infty} \no  \widetilde \chi_t^{(k)} - \chi_t^{(k)} \no_{\mathcal H_s^{(k),t} }^2 + \frac{e^{C^{\varphi_t}}}{N}
\end{align*}
by the same argument as in \eqref{ALPHA:ESTIMATE:FLUCTUATIONS}. Along the same steps, it follows that
\begin{align*}
\vert \eqref{DERIVATIVE:DIFFERENCE:CHI:BOG:7} \vert \le \sum_{k=0}^{\infty} \no  \widetilde \chi_t^{(k)} - \chi_t^{( k ) } \no_{\mathcal H_s^{(k),t} }^2 + \frac{e^{C^{\varphi_t}}}{N}.
\end{align*}
It remains to estimate
\begin{align*}
\eqref{DERIVATIVE:DIFFERENCE:CHI:BOG:8} \le &  C^{\varphi_t} N \Big( \no \wt  \chi_t^{( N-1 )} \no_{\mathcal H_s^{(N-1),t} } \no \chi_t^{ ( N+1 ) } \no_{\mathcal H_s^{(N+1),t} } + \no \wt \chi_t^{ ( N ) } \no_{\mathcal H_s^{(N),t} } \no \chi_t^{ ( N+2 ) } \no_{\mathcal H_s^{(N+2),t} } \Big) \\
\le & C^{\varphi_t} N^2\Big( \no \wt  \chi_t^{ ( N-1 ) }\no^2_{\mathcal H_s^{(N-1),t}}   + \no \wt \chi_t^{ ( N ) }\no^2_{\mathcal H_s^{(N),t} }  \Big) + \no \chi_t^{  ( N+1 ) } \no^2_{\mathcal H_s^{(N+1),t} } + \no \chi_t^{ ( N+2 ) } \no^2_{\mathcal H_s^{(N+2),t} }\\
\le   & C^{\varphi_t} N^2\Big( \no \wt  \chi_t^{ ( N-1 ) }\no^2_{\mathcal H_s^{(N-1),t}} + \no \wt \chi_t^{ ( N ) }\no^2_{\mathcal H_s^{(N),t} } \Big) +  \sum_{k=0}^\infty \no \wt  \chi_t^{ ( k ) } - \chi_t^{ ( k ) } \no^2_{\mathcal H_s^{(k),t} } \\
\le & \frac{e^{C^{\varphi_t}}}{N} + \sum_{k=0}^\infty \no \wt \chi_t^{ ( k ) } - \chi_t^{ ( k ) } \no^2_{\mathcal H_s^{(k),t} },
\end{align*}
where the last estimate follows from 
\begin{align*}
N^3 \Big(\no \wt  \chi_t^{( N-1 ) }\no^2_{\mathcal H_s^{(N-1),t} } + \no \wt \chi_t^{ ( N ) }\no^2_{\mathcal H^{(N),t} } \Big) \le \sum_{k=0}^{N} k^3 \no \wt \chi^{ ( k ) }_t \no_{\mathcal H_s^{(k),t} }^2  = N^3 \lsp \widetilde \Psi_t, (\cf{m}{t})^3 \widetilde \Psi_t\rsp \le e^{C^{\varphi_t}}.
\end{align*}
Hence, via Gr\"onwall,
\begin{align*}
 \sum_{k=0}^\infty \no \wt \chi_t^{(k)} - \chi_t^{( k ) } \no^2_{\mathcal H_s^{(k),t} } \le e^{C^{\varphi_t}} \Big( \sum_{k=0}^\infty \no \widetilde \chi_0^{( k ) }- \chi_0^{ ( k ) } \no^2_{\mathcal H_s^{(k),t} } + \frac{1}{N} \Big) = \frac{e^{C^{\varphi_t}}}{N}
\end{align*}
since $\widetilde \chi_0^{ (k) } = \chi_0^{ ( k ) }$ for all $k$. This completes the argument.
\end{proof}

\bigskip

\noindent{\it Acknowledgments.} We would like to thank Maximilian Jeblick, Phan Th\`{a}nh Nam and Marcin Napi\'{o}rkowski for helpful discussions.\ D.~M.\ gratefully acknowledges financial support from the German Academic Scholarship Foundation.\ S.~P.'s research has received funding from the People Programme (Marie Curie Actions) of the European Union's Seventh Framework Programme (FP7/2007-2013) under REA grant agreement n\textdegree~291734.

\end{document}